\documentclass[twocolumn,english,prl,superscriptaddress,longbibliography]{revtex4}
\usepackage{amssymb}
\usepackage{graphicx}
\usepackage{graphics}
\usepackage{amsmath}
\usepackage{amsthm}
\usepackage{color}
\usepackage{bbm}

\newcommand{\bea}{\begin{eqnarray}}
\newcommand{\eea}{\end{eqnarray}}

\def\bi{\begin{itemize}}
\def\ei{\end{itemize}}
\def\bc{\begin{center}}
\def\ec{\end{center}}

\def\C{\hbox{$\mit I$\kern-.7em$\mit C$}}
\def\R{\hbox{$\mit I$\kern-.6em$\mit R$}}

\def\ket#1{|#1\rangle}
\newcommand{\one}{\mbox{$1 \hspace{-1.0mm}  {\bf l}$}}
\def\tr{\mathrm{tr}}
\def\ket#1{\left| #1\right>}
\def\bra#1{\left< #1\right|}

\def\conv{\textrm{conv}}

\newtheorem{theorem}{Theorem}

\newtheorem{lemma}{Lemma}
\newtheorem{proposition}{Proposition}

\newtheorem{observation}{Observation}

\newcommand{\uno}{1\!\!1}

\begin{document}

\title{Asymptotic survival of genuine multipartite entanglement in noisy quantum networks depends on the topology}
\author{Patricia Contreras-Tejada}
\affiliation{Instituto de Ciencias Matem\'aticas, E-28049 Madrid, Spain}
\author{Carlos Palazuelos}
\affiliation{Departamento de An\'alisis Matem\'atico y Matem\'atica Aplicada, Universidad Complutense de Madrid, E-28040 Madrid, Spain}
\affiliation{Instituto de Ciencias Matem\'aticas, E-28049 Madrid, Spain}
\author{Julio I. de Vicente}
\affiliation{Departamento de Matem\'aticas, Universidad Carlos III de
Madrid, E-28911, Legan\'es (Madrid), Spain}

\begin{abstract}
The study of entanglement in multipartite quantum states plays a major role in quantum information theory and genuine multipartite entanglement signals one of its strongest forms for applications. However, its characterization for general (mixed) states is a highly nontrivial problem. We introduce a particularly simple subclass of multipartite states, which we term \textit{pair-entangled network} (PEN) states, as those that can be created by distributing exclusively bipartite entanglement in a connected network. We show that genuine multipartite entanglement in a PEN state depends on both the level of noise and the network topology and, in sharp contrast to the case of pure states, it is not guaranteed by the mere distribution of mixed bipartite entangled states. Our main result is a markedly drastic feature of this phenomenon: the amount of connectivity in the network determines whether genuine multipartite entanglement is robust to noise for any system size or whether it is completely washed out under the slightest form of noise for a sufficiently large number of parties. This latter case implies fundamental limitations for the application of certain networks in realistic scenarios, where the presence of some form of noise is unavoidable. To illustrate the applicability of PEN states to study the complex phenomenology behind multipartite entanglement, we also use them to prove superactivation of genuine multipartite nonlocality for any number of parties.
\end{abstract}

\maketitle

\paragraph{Introduction.}
Entanglement is at the core of the foundations of quantum mechanics and it is a crucial resource for the applications of quantum information theory \cite{reviewe}. The analysis of many-body entanglement has provided relevant tools for condensed matter physics \cite{reviewcm} and has given rise to several concrete multipartite applications such as secret sharing \cite{secretsharing}, conference key agreement \cite{conferencekey} and measurement-based quantum computation \cite{1wayqc}. Studying the complex ways in which multipartite entanglement manifests itself is thus interesting both theoretically and to come up with new applications. Of particular interest is the class of genuine multipartite entangled (GME) states, which are those that cannot be obtained by mixing only partially separable states and, therefore, entanglement spreads among all parties and not just a subset. GME is known to play a nontrivial role in certain quantum algorithms \cite{qalgo} and multipartite quantum key distribution schemes \cite{conferencekeygme}. Moreover, remarkably, it has been shown to be a necessary condition to achieve maximum sensitivity in quantum metrology \cite{gmesensing} and to obtain a multipartite private state and, hence, to establish a secret key \cite{GMEkey}. Thus, the certification of GME states has been studied in detail \cite{gmecertification}, although the characterization of entanglement in general is known to be computationally hard \cite{gurvits}. On the other hand, the preparation, control and distribution of GME states is a major experimental challenge. Arguably, the most feasible way to achieve this (e.g.\ in quantum optics applications) is by distributing exclusively bipartite entanglement among different pairs of parties giving rise to a connected network. In fact, quantum networks are currently being actively investigated as a realistic platform for quantum information processing. This includes establishing long-range entanglement starting from smaller entanglement links or harnessing node-to-node entanglement in order to achieve on-demand quantum communication between different possible subsets of parties (see \cite{reviewn} and references therein).

In this Letter we intend to put forward a theoretical analysis, from the point of view of entanglement theory, of the properties of the underlying states that arise in quantum networks, which we term \textit{pair-entangled network} (PEN) states. It should be noticed that such a state is a universal resource under local operations and classical communication (LOCC) provided that the underlying graph is connected and the bipartite entanglement shared by the nodes is of sufficient quality. This is because by means of local preparation and teleportation the parties can then end up sharing \textit{any} quantum state of a given local dimension. Recent work has considered the limitations arising from the distribution of arbitrary bipartite entanglement in networks when state manipulation is bound to a class of operations that is a strict subset of LOCC and it has been shown that certain GME states cannot be prepared in this way \cite{GNME1,GNME2}. However, here we study the entanglement properties of PEN states within the LOCC paradigm depending on the type of entanglement shared and the topology of the network \cite{networkLOCC}. From this point of view, to ask when a PEN state is GME seems a very relevant question as this makes it possible to benchmark the quality of the quantum network. If the corresponding PEN state is not GME, then it cannot be transformed by LOCC into a GME state and, therefore, as pointed out above, this leads to fundamental limitations in applications.

A simple argument shows that all pure PEN states are GME independently of the amount of entanglement shared and the geometry of the network (if it is connected) and, actually, we have recently shown that they are even genuine multipartite nonlocal (GMNL) \cite{networkGMNL}. However, in realistic implementations noise is unavoidable and mixed PEN states must be considered.
The previous property of pure PEN states directly implies that, for a fixed network, GME should be robust to some noise, but the extent to which this holds is unclear. In fact, to our knowledge it is not known whether sharing arbitrary bipartite entanglement is enough to guarantee that a PEN state is GME. Here, we consider a simple and realistic model in which the nodes share isotropic states, i.e. maximally entangled states mixed with white noise, and show that the answer to the above question is negative. The mere fact that the nodes share bipartite entanglement does not imply that a connected network is GME. Furthermore, this not only depends on the level of noise but also on the topology of the network. However, our main result is a more extreme feature of this phenomenon. Instead of asking for which value of the noise parameter a given network is GME, we consider a more realistic approach in which the noise parameter is fixed and we ask which networks display GME under this constraint. It turns out that, for any nonzero value of the noise, any tree network of sufficiently many parties is no longer GME and, on the contrary, the GME of a completely connected network persists for any number of parties if the noise is below some threshold. Thus, asymptotic survival of GME depends drastically on the geometry of the network. While any unavoidable limitation in the ability to prepare entangled states upper bounds the number of parties that can share GME in some topologies, a larger connectivity guarantees GME for any size provided that a certain level of quality in the prepared entangled states can be achieved.

In addition to this, the overwhelming complexity of multipartite state space has often led to constrain the study of entanglement to subsets of states with relevant physical and/or mathematical properties such as graph states \cite{graph}, locally maximally entangleable states \cite{LME} or tensor network states \cite{tensor}. We believe that the class of PEN states is a promising platform endowed with a clear operational motivation in order to study the rich phenomenology of multipartite entanglement. Based on this, using PEN states we provide examples of GME states which are not GMNL, different from those known before \cite{GMEneqGMNL1,GMEneqGMNL2}. The tensor product structure of quantum theory enables the fact that objects that are not resourceful may display this resource when several copies of them are taken together, a phenomenon known as superactivation. This is the case of nonlocality, where it has been proven that superactivation is possible in the bipartite case \cite{palazuelos_super-activation_2012}. Here, building on this construction, we prove superactivation of GMNL for any number of parties using the aforementioned states.

\paragraph{Preliminaries.}
As mentioned above, we will consider networks where the nodes share isotropic states on $\mathbb{C}^d\otimes\mathbb{C}^d$:
\begin{equation}\label{eqisotropic}
\rho(p)=p\phi^+_d+(1-p)\tilde{\mathbbm{1}},
\end{equation}
where $|\phi^+_d\rangle=(1/\sqrt{d})\sum_{i=0}^{d-1}|ii\rangle$ is the $d$-dimensional maximally entangled state, $\phi_d^+=|\phi_d^+\rangle\langle\phi_d^+|$ and $\tilde{\mathbbm{1}}=\mathbbm{1}/d^2$. Isotropic states not only represent a standard noise model but they also possess nontrivial symmetry properties. This has led to an in-depth study of these states and they appear as an intermediate step in several protocols \cite{isotropic}. In particular, isotropic states are entangled if and only if $p>1/(d+1)$ \cite{isotropic}.

PEN states are defined by selecting an undirected graph $G=(V,E)$ that encodes the structure of the network. The vertices $V=[n]:=\{1,2,\ldots,n\}$ represent the parties and the edges $E\subseteq\{(i,j):i,j\in V, i<j\}$ represent when two nodes share a bipartite state. In order to specify the PEN state, one must specify $G$ as well as which state is associated to every edge in $E$. In our case, we will always consider isotropic states $\rho_{ij}(p)$ shared by parties $i$ and $j$ \footnote{We sometimes also label the parties who share a state by the edge instead of the vertices, i.e.\ if $e=(i,j)$, then $\rho_e(p)=\rho_{ij}(p)$.} and, for simplicity, we will often consider that all edges are given by the same isotropic state. Thus, given the graph $G$ and the noise parameter $p$, the corresponding isotropic PEN state is
\begin{equation}\label{pen}
\sigma_G(p)=\bigotimes_{(i,j)\in E}\rho_{ij}(p).
\end{equation}
Here, the indices in the tensor product indicate to which local Hilbert space each qudit of the isotropic states belongs \cite{suppmat}. Thus, given $G$, party $i$ holds $\deg(i)$ qudits (where $\deg(i)$ is the degree of vertex $i$) and the local dimension of $\sigma_G(p)$ for each party $i$ is $d^{\deg(i)}$. We will focus on some particular graphs: A tree graph is a graph with no cycles, such as the star graph in which a central node is connected to all other vertices and there are no more edges. These are graphs with the lowest connectivity. On the other hand, a completely connected graph is that for which $E=\{(i,j):i,j\in V, i<j\}$. Sometimes it will be convenient to alter the notation for vertices in order to label the different particles held by one party. For instance, for 3 parties $A$, $B$ and $C$ the star and completely connected PEN states can be also respectively denoted by
\begin{align}\label{3pen}
\sigma_\textnormal{star}(p)&=\rho_{A_1B}(p)\otimes\rho_{A_2C}(p),\nonumber\\
 \sigma_\textnormal{cc}(p)&=\rho_{A_1B_1}(p)\otimes\rho_{A_2C_1}(p)\otimes\rho_{B_2C_2}(p).
\end{align}

Last, we provide the definition of GME. Given the $n$-partite Hilbert space $H=\bigotimes_{i=1}^nH_i$, a pure state $|\psi\rangle\in H$ is biseparable (otherwise GME) if $|\psi\rangle=|\psi_M\rangle\otimes|\psi_{\overline M}\rangle$ for some $M\subsetneq[n]$ and its complement $\overline M$, where $|\psi_M\rangle\in \bigotimes_{i\in M}H_i$ and $|\psi_{\overline M}\rangle\in \bigotimes_{i\in\overline M}H_{i}$. The definition extends to mixed states by taking the convex hull: the set of biseparable states is $\conv\{|\psi\rangle\langle\psi|:|\psi\rangle\textrm{ is biseparable}\}$ and a state that does not belong to it is GME. It follows from this definition that the set of biseparable states is closed under LOCC. Notice that, for PEN states, it is immediate that if a subset of the network only shares separable states with its complement, then the PEN state is biseparable. It is worth pointing out that studying GME in PEN states built from isotropic states is not only a standard noise model but also quite general, since all states with entangled fraction larger than $1/d$ can be transformed by LOCC into an entangled isotropic state \cite{isotropic}.

\paragraph{Robustness of GME for isotropic PEN states.}
The fact that any PEN state is GME for any connected network sharing arbitrary bipartite pure entangled states follows by noticing that the reduced state corresponding to any subset of parties $M\subsetneq[n]$ will in this case be mixed. Notice that, since the set of biseparable states is closed, any given fixed PEN state will then tolerate some noise in its edges so as to remain GME. However, this still leaves open the question of whether sharing arbitrary bipartite entanglement in any connected network is enough to generate GME. We start by observing that already the simplest case of tripartite PEN states with 2-qubit isotropic edges (cf.\ Eq.\ (\ref{3pen})) shows that this is not the case (disproving moreover a conjecture in \cite{sunchen}). Although we did not compute the exact thresholds, by explicitly constructing biseparable decompositions and using the techniques of \cite{jungnitsch} to build fully decomposable witnesses for these states, in \cite{suppmat} we prove bounds on the noise parameter $p$ that guarantee biseparability or GME for $\sigma_{\textnormal{star}}(p)$ and $\sigma_{\textnormal{cc}}(p)$. The results are summarized in Table \ref{tab:lambdatriangle}. Notice that for $0.491<p\leq0.547$, $\sigma_\textnormal{cc}(p)$ is GME while $\sigma_{\textnormal{star}}(p)$ is biseparable. Thus, this proves the intuitive fact that increasing the connectivity by producing more links makes GME more robust to noise.

\begin{table}[ht]
\centering %
\begin{tabular}{c|cc}
 & biseparable for $p\leq$ & GME for $p>$\tabularnewline
\hline
$\sigma_\textnormal{star}(p)$ & $(1+2\sqrt{2})/7\simeq0.547$ & $1/\sqrt{3}\simeq0.577$\tabularnewline
$\sigma_\textnormal{cc}(p)$ & $3/7\simeq0.429$ & $(2\sqrt{5}-3)/3\simeq0.491$\tabularnewline
\end{tabular}\caption{Bounds for biseparability and GME for tripartite PEN states with 2-qubit isotropic edges. Notice that both states can be biseparable above the threshold $p>1/3$ that determines that the edges are entangled.}
\label{tab:lambdatriangle}
\end{table}

We now move onto our main result. The above observations show that if an experimental implementation is bound to a certain visibility in the preparation of isotropic states, the ability to display GME may depend on the network configuration. Nevertheless, improving the apparatuses to produce isotropic states with $p>0.577$ will suffice in any connected 3-partite configuration. Increasing this visibility will ensure GME for any network of a \textit{fixed} number of parties. However, this does not necessarily imply that GME asymptotically survives, i.e.\ that there is a threshold in the visibility an experimentalist can aim at above which GME is guaranteed independently of the number of parties. This is indeed a more realistic situation, where a certain quantum state can be obtained in experiments and one wants to use it in a large network. Since deleting edges is LOCC (as this amounts to tracing out subsystems), asymptotic survival of GME in one configuration ensures it for those with more links; however, it is not at all clear in principle whether this phenomenon is universal, impossible or whether it depends on the network.

We first focus on a general class of PEN states which covers the networks of lowest connectivity: tree graphs. For this family, we find a negative answer to the question of asymptotic survival of GME.
\begin{theorem}\label{thm1}
Let $G=(V,E)$ be a tree graph with $n$ vertices and let $\sigma_{G}(p)$ denote the corresponding $n$-partite isotropic PEN state as given by Eq.\ (\ref{pen}). Then, $\sigma_{G}(p)$ is biseparable if $|E|\geq dp/(1-p)$.
\end{theorem}

The proof (cf.\ \cite{suppmat}) is obtained by obtaining an explicit biseparable decomposition of $\sigma_{G}(p)$ using the separability properties of isotropic states.
To illustrate the previous result, note that the $n$-partite PEN state $\sigma_\textnormal{star}(p)$ with 2-qubit isotropic edges is biseparable when $n\geq(1+p)/(1-p)$. Thus, a visibility $p=0.6$ precludes GME for more than 3 parties, while the already experimentally demanding value of $p=0.95$ bounds the size to 38 parties. This shows a fundamental limitation to GME distribution in practical scenarios such as the star configuration in which a powerful central laboratory prepares entangled states for satellite nodes.

It should be stressed that the proof of Theorem \ref{thm1} can be easily generalized to other noise models and, more importantly, to other networks. In this sense, in \cite[Theorem 4]{suppmat} we prove a similar result for polygonal networks, i.e. those based on a cycle graph. At this point one may wonder whether asymptotic survival of GME is at all possible. Our next result shows that this is indeed the case by considering the network of highest connectivity.
\begin{theorem}\label{thm2}
Let $G$ be a completely connected graph of $n$ vertices and let $\sigma_\textnormal{cc}(p)$ denote the corresponding $n$-partite isotropic PEN state as given by Eq.\ (\ref{pen}). Then, there exists a value of $p_0<1$, which is independent of $n$ (i.e.\ depends only on $d$), such that $\sigma_\textnormal{cc}(p)$ is GME for every $n$ and for all $p>p_0$.
\end{theorem}


Hence, our results uncover a fundamental property of entanglement in quantum networks: asymptotic survival of GME depends on the topology. The proof of Theorem \ref{thm2}, which is given in \cite{suppmat}, relies on two parts. First, we establish an upper bound on the sum over all pairs of parties of the fidelity with the maximally entangled state $\phi_2^+$ that can be achieved after any LOCC protocol starting with a biseparable state. Then, we show that, above a certain threshold in the visibility, $\sigma_\textnormal{cc}(p)$ can overcome this bound by edge teleportation and entanglement distillation when $n$ is large. Once GME is ensured to persist for a large number of parties, it follows that, for a fixed, large enough visibility, GME can be guaranteed for completely connected networks of any size. The precise value of the threshold $p_0$ can be explicitly given (at least in the limit of large size) as this is controlled by the success of the particular entanglement distillation protocol that is implemented. We used the one-way distillation protocol of \cite{1waydistillation}, which in the particular case where the nodes share 2-qubit isotropic states yields $p_0\simeq0.865$. We did not attempt any optimization in this direction.

\paragraph{Constructing PEN states with relevant entanglement properties}

In addition to the relevance of PEN states in the context of networks, we find this family extremely versatile to study general properties of multipartite entanglement. Here we will focus on the relation between quantum entanglement and nonlocality. The latter concept refers to the possibility of obtaining certain correlations when performing separate measurements on multipartite quantum states which cannot be explained classically, and it is crucial in many applications in quantum information theory \cite{reviewnl}. In precise terms, a given $n$-partite probability distribution $P=\left\{ P(\alpha_{1}\alpha_{2}...\alpha_{n}|\chi_{1}\chi_{2}...\chi_{n})\right\} _{\alpha_{1},...,\alpha_{n},\chi_{1},...,\chi_{n}}$
(with input $\chi_{i}$ and output $\alpha_{i}$ for party $i$) is said to be GMNL if it is not of the form
\begin{equation}
\begin{aligned}P(\alpha_{1}\alpha_{2}...\alpha_{n} & |\chi_{1}\chi_{2}...\chi_{n})\\
=\sum_{M\subsetneq[n]}\sum_{\lambda} & q_{M}(\lambda)P_{M}(\{\alpha_{i}\}_{i\in M}|\{\chi_{i}\}_{i\in M},\lambda)\\
 & \times P_{\overline{M}}(\{\alpha_{i}\}_{i\in\overline{M}}|\{\chi_{i}\}_{i\in\overline{M}},\lambda),
\end{aligned}
\label{eq:BLdistrNpartite}
\end{equation}
where $q_{M}(\lambda)\geq0\,\forall\lambda,M$ and $\sum_{\lambda,M}q_{M}(\lambda)=1$. Otherwise, we say that  $P$ is bilocal. The distributions $P_{M},\,P_{\overline{M}}$ will be assumed to be nonsignalling as this captures most physical situations better than unrestricted $P_{M},\,P_{\overline{M}}$ \cite{GWAM12,BBGS13,SRB20, WSSKS20}. An $n$-partite state $\rho$ is GMNL if local measurements $\{E^{(i)}_{\alpha_i|\chi_i}\geq0\}$ ($\sum_{\alpha_i}E^{(i)}_{\alpha_i|\chi_i}=\mathbbm{1}$ $\forall\chi_i,i$) exist which give rise to a GMNL distribution
\begin{equation}
P(\alpha_{1}\alpha_{2}...\alpha_{n} |\chi_{1}\chi_{2}...\chi_{n})=\tr(\rho\bigotimes_{i=1}^nE^{(i)}_{\alpha_i|\chi_i}).
\end{equation}

While GMNL states are GME, as mentioned in the introduction, the converse implication is not true for any number of parties \cite{GMEneqGMNL1,GMEneqGMNL2}. Finding more examples of GME states that are bilocal and the conditions under which this might happen is crucial to fully understand the relation between entanglement and nonlocality in the multipartite setting.  In fact, the first such example found in the bipartite case \cite{Werner89, Barrett02} is a cornerstone in the field.

It is worth mentioning that many copies of isotropic PEN states are always GME (as long  as $p>1/(d+1)$ for the underlying isotropic states). The fact that this holds even for biseparable PEN states is possible because the set of biseparable states is not closed under tensor products. Indeed, taking many copies of an isotropic PEN state can be understood as having another PEN state with the same topology but where each edge represents many copies of an isotropic state and, thus, whose edges are more entangled. By means of the LOCC protocol of \cite{isotropic}, starting from sufficiently many copies of any isotropic PEN state one can distill another PEN state where each edge represents a state arbitrarily close to a maximally entangled state. However, this new state is GME (in fact, GMNL by \cite{networkGMNL}) and, therefore, the original state must be GME as well.

The situation is not so clear when looking at nonlocality since LOCC transformations do not preserve the set of bilocal states. Being able to obtain a GMNL state by taking many copies of a bilocal one would yield GMNL superactivation. While this has been shown in the bipartite scenario \cite{palazuelos_super-activation_2012}, to our knowledge, it has not been studied for more than two parties. Our last result tackles the previous two questions. It provides new families of bilocal GME states and, moreover, it shows that superactivation can also hold in the multipartite setting.
\begin{theorem}\label{thm superactiv}
Let $\tau(p)$ denote the $n$-partite PEN state corresponding to a star graph in which all edges represent the maximally entangled state except one, which is given by the isotropic state $\rho(p)$. Then, if
\begin{equation}\label{nonsteerablem}
\frac{1}{d+1}<p\leq\frac{(3d-1)(d-1)^{d-1}}{(d+1)d^d},
\end{equation}
\begin{itemize}
\item[(i)] $\tau(p)$ is GME $\forall n\geq3$.
\item[(ii)] $\tau(p)$ is not GMNL $\forall n\geq3$.
\item[(iii)] $\tau(p)^{\otimes k}$ is GMNL $\forall n\geq3$ if $k$ is large enough.
\end{itemize}
\end{theorem}

To obtain this result (see \cite{suppmat}), we prove that any star network with an entangled isotropic state on one edge and maximally entangled states on the rest is GME. We also establish a connection between having bilocality of PEN states and the edges being nonsteerable---a well-studied property of bipartite quantum states, intermediate between entanglement and nonlocality \cite{steering}. We show that any star network with a nonsteerable state on one edge is automatically bilocal. Combining the previous two results we can obtain a network $\tau(p)$ verifying conditions (i) and (ii) above, where the bounds in Eq.\ (\ref{nonsteerablem}) guarantee that the edge with the isotropic state is entangled but nonsteerable \cite{almeida}. Finally, using the ideas of \cite{amr_unbounded_2020}, we extend the Bell inequality used to prove bipartite superactivation in \cite{palazuelos_super-activation_2012,cavalcanti_all_2013} to a multipartite one in order to show that $\tau(p)^{\otimes k}$ is GMNL for a large enough $k$.

\paragraph{Conclusions.}
In this Letter we have introduced the class of multipartite PEN states as those underlying the current proposals of quantum networks and we have investigated their GME properties. We have shown that sharing bipartite entanglement in a connected network does not guarantee GME, but that both a higher quality of node-to-node entanglement and a larger connectivity play in favour of displaying this property. Our main result is a drastically contrasting behaviour with respect to this feature: While tree isotropic PEN states cannot be GME for any value of the visibility $p<1$ for sufficiently many parties, the GME of the completely connected PEN state is robust for all visibilities above a fixed threshold for any system size. Furthermore, the class of PEN states is an operationally motivated subset of multipartite states with a clear mathematical structure in which the well-developed theory of bipartite entanglement can be exploited to analyze entanglement in the multipartite scenario. Thus, we have provided a construction of GME but non-GMNL PEN states for any number of parties that lead to superactivation of GMNL.

Besides these particular results, we believe that PEN states might find applications in different contexts and that this work can be continued in several directions. We conclude by posing two such possibilities. First, tree graphs and the completely connected graph represent the two most extreme cases in terms of connectivity. What is the minimal amount of connectivity that enables asymptotic survival of GME? Second, the aforementioned property of tree networks implies that their GMNL cannot asymptotically survive either. However, can the asymptotic survival of GME in completely connected PEN states be extended to GMNL? The dependence of these features on the geometry of the network suggests that there might be a fruitful interplay between these problems and the theory of complex networks.

\begin{acknowledgments}
The authors thank Andreas Winter and David Elkouss for enlightening discussions. This research was funded by the Spanish Ministerio de Ciencia e Innovaci\'{o}n through Grant No. PID2020-113523GB-I00 funded by MCIN/AEI/ 10.13039/501100011033 and by ``ERDF A way of making Europe", and by the Comunidad de Madrid through Grant No. QUITEMAD-CMS2018/TCS-4342. We also acknowledge funding from the Spanish MINECO, through the ``Severo Ochoa Programme
for Centres of Excellence in R\&D'' SEV-2015-0554 funded by MCIN/AEI/ 10.13039/501100011033 and from the Spanish
National Research Council, through the ``Ayuda extraordinaria a Centros
de Excelencia Severo Ochoa'' 20205CEX001 (PCT and CP) and Comunidad de Madrid under the Multiannual Agreement with UC3M in the line of Excellence of University Professors (EPUC3M23), and in the context of the V PRICIT (Regional Programme of Research and Technological Innovation)(JIdV).
\end{acknowledgments}

\begin{widetext}

\section{Supplemental material}

In this supplemental material we prove the content of Table I and Theorems 1, 2 and 3 in the main text. We recall that, in the PEN states we consider, all nodes connected by an edge are assumed to share an isotropic state with visibility $p$
\begin{equation*}
\rho(p)=p\phi^+_d+(1-p)\tilde{\uno},
\end{equation*}
where $|\phi^+_d\rangle=(1/\sqrt{d})\sum_{i=0}^{d-1}|ii\rangle$ is the $d$-dimensional maximally entangled state, $\phi_d^+=|\phi_d^+\rangle\langle\phi_d^+|$ and $\tilde{\one}=\uno/d^2$. As stated in the main text, $\rho(p)$ is entangled iff $p>1/(d+1)$ \cite{isotropic}. When the local dimension $d$ is clear from the context, we write the maximally entangled state simply as $\phi^+$ in order to leave room in the subscript for specifying which parties share it.

\subsection{GME in isotropic PEN states of three parties}

The only possible connected graphs for the case of three parties ($A$, $B$ and $C$) are the star and the completely connected graph. Hence, as explained in the main text, in this section we consider the PEN states
\begin{align*}
\sigma_{\textnormal{star}}(p)&=\rho_{A_1B}(p)\otimes\rho_{A_2C}(p),\nonumber\\
\sigma_{\textnormal{cc}}(p)&=\rho_{A_1B_1}(p)\otimes\rho_{A_2C_1}(p)\otimes\rho_{B_2C_2}(p),
\end{align*}
corresponding to the two possible networks mentioned above, in which the edges stand for 2-qubit isotropic states of the same visibility $p$.

Before presenting the results of this section, let us write the state $\sigma_{\textnormal{star}}(p)$ for isotropic states of local dimension 2 in full detail so that its description is completely clear to the reader (we could analogously write $\sigma_{\textnormal{cc}}(p)$). First of all, note that the state space of the tripartite system is $\mathbb C^2_{A_1}\otimes C^2_{A_2}\mathbb \otimes \mathbb C_B^2\otimes \mathbb C_C^2$, where registers $A_1$ and $A_2$ belong to the same party $A$ (so it has dimension 4), while registers $B$ and $C$, each of dimension 2, belong to the other two parties. Then, if we denote $\phi^+=|\phi ^+\rangle\langle\phi^+|$ and $\tilde{\one}=\uno/4$, where $|\phi^+\rangle=(1/\sqrt{2})(|00\rangle+|11\rangle)$, we have
\begin{align*}
\sigma_{\textnormal{star}}(p)&=(p\phi_{A_1B}^++(1-p)\tilde{\uno}_{A_1B})\otimes (p\phi_{A_2C}^++(1-p)\tilde{\uno}_{A_2C})\\&=p^2(\phi_{A_1B}^+\otimes \phi_{A_2C}^+)+p(1-p)(\phi_{A_1B}^+\otimes\tilde{\uno}_{A_2C})\\&+
p(1-p)(\tilde{\uno}_{A_1B}\otimes\phi_{A_2C}^+)+ (1-p)^2(\tilde{\uno}_{A_1B}\otimes \tilde{\uno}_{A_2C}).
\end{align*}

We can write this state in complete detail by considering a basis $\{|0\rangle, |1\rangle\}$ of $\mathbb C^2$. Then, we have
\begin{align*}
\sigma_{\textnormal{star}}(p)&=p^2\frac{1}{4}\sum_{i,j,i',j'}|i_{A_1}i'_{A_2}\rangle\langle j_{A_1}j'_{A_2}|\otimes |i_{B}\rangle \langle j_B|\otimes |i'_{C}\rangle \langle j'_C|\\&+p(1-p)\frac{1}{8}\sum_{i,j,i',j'}|i_{A_1}i'_{A_2}\rangle\langle j_{A_1}i'_{A_2}|\otimes |i_{B}\rangle \langle j_B|\otimes |j'_{C}\rangle \langle j'_C|
\\&+p(1-p)\frac{1}{8}\sum_{i,j,i',j'}|i_{A_1}i'_{A_2}\rangle\langle i_{A_1}j'_{A_2}|\otimes |j_{B}\rangle \langle j_B|\otimes |i'_{C}\rangle \langle j'_C|\\&+(1-p)^2\frac{1}{16}\sum_{i,j,i',j'}|i_{A_1}i'_{A_2}\rangle\langle i_{A_1}i'_{A_2}|\otimes |j_{B}\rangle \langle j_B|\otimes |j'_{C}\rangle \langle j'_C|,
\end{align*}where the sums run in $i,j,i',j'=0,1$. Note that, in particular, the first term in the sum above can be seen as a maximally entangled state in $A_1A_2|BC$ and the last term is nothing but the (normalized) identity in $A_1A_2BC$.

The goal of this section is to prove the biseparability and GME bounds provided in Table I in the main text, which we reproduce here for the reader's convenience as Table II. Each entry of the table is accounted for by one of the following four propositions. The proofs of the biseparability bounds do not require the assumption that $d=2$ and, hence, we consider the more general case of arbitrary local dimension $d$ for the isotropic states underlying the network. In the proofs we use the following additional notation. The (2-dimensional) flip or swap operator is denoted by $\Pi=\sum_{i,j=0}^{1}\ket{ij}\bra{ji}$. The superscript $\Gamma_M$ stands for partial transposition with respect to the parties in $M\subsetneq[n]$ and $X\succcurlyeq0$ means that the Hermitian matrix $X$ is positive semidefinite.

\begin{table}[ht]
\centering %
\begin{tabular}{c|cc}
 & biseparable for $p\leq$ & GME for $p>$\tabularnewline
\hline
$\sigma_{\textnormal{star}}(p)$ & $(1+2\sqrt{2})/7\simeq0.547$ & $1/\sqrt{3}\simeq0.577$\tabularnewline
$\sigma_{\textnormal{cc}}(p)$ & $3/7\simeq0.429$ & $(2\sqrt{5}-3)/3\simeq0.491$\tabularnewline
\end{tabular}\caption{Bounds for biseparability and GME for 3-partite 2-qubit-sharing isotropic PEN states. Notice that both states can be biseparable above the threshold $p>1/3$ that determines that the edges are entangled.}
\end{table}

\begin{proposition} \label{thm:lambda-iso-gme}
The 3-partite PEN state $\sigma_{\textnormal{star}}(p)$ with 2-qubit isotropic states is GME if $p>1/\sqrt{3}$.
\end{proposition}

\begin{proof}
Throughout this proof we omit the subscripts labeling the qubits in all operators: $X\otimes Y:=X_{A_1B}\otimes Y_{A_2C}$, following the same order as in Eq.\ (\ref{3pen}). We will first show that the operator
\begin{equation}
W=\one\otimes\one+2\one\otimes\phi^{+}+2\phi^{+}\otimes\one-8\phi^{+}\otimes\phi^{+}\label{eq:witpi-1}
\end{equation}
is a GME witness. In order to do so, it suffices to prove that $\tr(W\rho)\geq0$
for every $\rho$ that is a PPT mixture \cite{jungnitsch}. In turn, for this it is enough
to see that there exist $P_{M},Q_{M}\succcurlyeq0$ such that $W=P_{M}+Q_{M}^{\Gamma_{M}}$
for $M=A,B,C$ \cite{jungnitsch}. It is straightforward
to verify that this is indeed the case (we only need to use that $\Pi^{\Gamma}=2\phi^{+}$ for the 2-dimensional maximally entangled state)
if:

\begin{align}
P_{A} & =2\phi^{+}\otimes(\one-\phi^{+})+2(\one-\phi^{+})\otimes\phi^{+},\\
Q_{A} & =\frac{1}{2}[(\one-\Pi)\otimes(\one+\Pi)+(\one+\Pi)\otimes(\one-\Pi)],\\
P_{B} & =0,\quad Q_{B}=(\one+\Pi)\otimes(\one-\phi^{+})+3(\one-\Pi)\otimes\phi^{+},\\
P_{C} & =0,\quad Q_{C}=(\one-\phi^{+})\otimes(\one+\Pi)+3\phi^{+}\otimes(\one-\Pi),
\end{align}
where we have used that $\phi^{+},\one-\phi^{+},\one\pm\Pi\succcurlyeq0$
and that the sum and tensor product of positive semidefinite matrices
is positive semidefinite.

Next, we only need to observe that
\begin{equation}
\tr(W\sigma_{\textnormal{star}}(p))=\frac{3}{2}\left(1-3p^2\right),
\end{equation}
which is strictly smaller than zero for the values of $p$ provided in the statement of the Proposition.
\end{proof}

\begin{proposition}\label{thm:triangle-GME}
The 3-partite PEN state $\sigma_{\textnormal{cc}}(p)$ with 2-qubit isotropic states is GME if $p>(2\sqrt{5}-3)/3$.
\end{proposition}

\begin{proof}
The proof strategy is the same as in Proposition \ref{thm:lambda-iso-gme} and, as therein, we also omit subscripts for operators: $X\otimes Y\otimes Z:=X_{A_1B_1}\otimes Y_{A_2C_1}\otimes Z_{B_2C_2}$, following the same order as in Eq.\ (\ref{3pen}), unless otherwise explicitly stated. We will show the operator
\begin{equation}
W= \one\otimes\one\otimes\phi^{+}+\one\otimes\phi^{+}\otimes\one+\phi^{+}\otimes\one\otimes\one-\one\otimes\phi^{+}\otimes\phi^{+}-\phi^{+}\otimes\phi^{+}\otimes\one-\phi^{+}\otimes\one\otimes\phi^{+}-3\phi^{+}\otimes\phi^{+}\otimes\phi^{+}
\end{equation}
can be decomposed as
\begin{equation}
W=P_{M}+Q_{M}^{\Gamma_{M}}
\end{equation}
for each bipartition $M=A,B,C$, where $P_{M},Q_{M}\succcurlyeq0$
for all $M.$ Indeed, we have
\begin{equation}
\begin{aligned}P_{A} & =\one\otimes\phi^{+}\otimes(\one-\phi^{+})+\phi^{+}\otimes(\one-\phi^{+})\otimes(\one-\phi^{+})\\
Q_{A} & =\frac{1}{2}[(\one-\Pi)\otimes(\one+\Pi)+(\one+\Pi)\otimes(\one-\Pi)]\otimes\phi^{+}\\
P_{B} & =\one\otimes(\one-\phi^{+})\otimes\phi^{+}+\phi^{+}\otimes(\one-\phi^{+})\otimes(\one-\phi^{+})\\
Q_{B} & =\frac{1}{2}[(\one-\Pi)\otimes(\one+\Pi)+(\one+\Pi)\otimes(\one-\Pi)]_{A_{1}B_{1}B_{2}C_{2}}\otimes\phi_{A_{2}C_{1}}^{+}\\
P_{C} & =(\one-\phi^{+})\otimes\phi^{+}\otimes\one+(\one-\phi^{+})\otimes(\one-\phi^{+})\otimes\phi^{+}\\
Q_{C} & =\phi^{+}\otimes\frac{1}{2}[(\one-\Pi)\otimes(\one+\Pi)+(\one+\Pi)\otimes(\one-\Pi)].
\end{aligned}
\end{equation}
Thus, $W$ is a GME witness. Last, one needs to notice that
\begin{equation}
\tr(W\sigma_{\textnormal{cc}}(p))=\frac{3}{64}(11+15p-63p^{2}-27p^{3}),
\end{equation}
which is strictly smaller than zero when $p>(2\sqrt{5}-3)/3$.
\end{proof}

\begin{proposition}\label{thm:lambda-iso-BS}
The 3-partite PEN state $\sigma_{\textnormal{star}}(p)$ with $d$-dimensional isotropic states is biseparable if
\begin{equation}
p\leq\frac{(1+\sqrt{2})d-1}{d^{2}+2d-1}.
\end{equation}
\end{proposition}

\begin{proof}
Notice that
\begin{equation}
\sigma_{\textnormal{star}}(p)=p^{2}  \phi_{A_{1}B}^{+}\otimes\phi_{A_{2}C}^{+}+p(1-p)\phi_{A_{1}B}^{+}\otimes\tilde{\one}_{A_{2}C}+p(1-p)\tilde{\one}_{A_{1}B}\otimes\phi_{A_{2}C}^{+}+(1-p)^{2}\tilde{\one}_{A_{1}B}\otimes\tilde{\one}_{A_{2}C},
\end{equation}
which can be rewritten as
\begin{equation}
\begin{aligned}\sigma_{\textnormal{star}}(p)&=(1-q)p^{2}\phi_{A_{1}B}^{+}\otimes\phi_{A_{2}C}^{+}+(1-p)^{2}\tilde{\one}_{A_{1}B}\otimes\tilde{\one}_{A_{2}C}\\
 & +\phi_{A_{1}B}^{+}\otimes\left(\frac{qp^{2}}{2}\phi_{A_{2}C}^{+}+p(1-p)\tilde{\one}_{A_{2}C}\right)+\left(\frac{qp^{2}}{2}\phi_{A_{1}B}^{+}+p(1-p)\tilde{\one}_{A_{1}B}\right)\otimes\phi_{A_{2}C}^{+}
\end{aligned}
\label{eq:lambdaisoBS}
\end{equation}
for any $q\in[0,1].$ Now, the first line can be seen as an (unnormalized)
isotropic state in parties $A|BC$ (with dimension $d^{2}$), while
the second line contains isotropic states in $A_{2}|C$ and $A_{1}|B$
respectively (with dimension $d$). Showing that each of these states
is separable will entail the result.

Denote these isotropic states by $\sigma_{0},\sigma_{1},\sigma_{2}$
in the order that they appear in equation (\ref{eq:lambdaisoBS}).
Normalizing $\sigma_{0},$ we find
\begin{equation}
\sigma_{0}=\frac{(1-q)p^{2}\phi_{A_{1}B}^{+}\otimes\phi_{A_{2}C}^{+}+(1-p)^{2}\tilde{\one}_{A_{1}B}\otimes\tilde{\one}_{A_{2}C}}{(1-q)p^{2}+(1-p)^{2}},
\end{equation}
meaning it is separable in $A|BC$ if
\begin{equation}
\frac{(1-q)p^{2}}{(1-q)p^{2}+(1-p)^{2}}\leq\frac{1}{d^{2}+1},
\end{equation}
i.e., if
\begin{equation}
q\geq1-\frac{(1-p)^{2}}{p^{2}d^{2}}.
\end{equation}
Normalizing $\sigma_{1},$ we obtain
\begin{equation}
\sigma_{1}=\frac{qp^{2}\phi_{A_{2}C}^{+}+2p(1-p)\tilde{\one}_{A_{2}C}}{qp^{2}+2p(1-p)},
\end{equation}
which is separable if
\begin{equation}
\frac{qp^{2}}{qp^{2}+2p(1-p)}\leq\frac{1}{d+1}.
\end{equation}
Simplifying, this entails that
\begin{equation}
q\leq\frac{2-2p}{pd}.
\end{equation}
Reasoning symmetically, the separability of $\sigma_{2}$ gives the
same bound. Both bounds on $q$ together entail that
\begin{equation}
1-\frac{(1-p)^{2}}{p^{2}d^{2}}\leq\frac{2-2p}{pd}
\end{equation}
and, solving for $p,$ we obtain the desired bound.
\end{proof}

\begin{proposition}\label{thm:triangle-BS}
The 3-partite PEN state $\sigma_{\textnormal{cc}}(p)$ with $d$-dimensional isotropic states is biseparable if  $p\leq3/(3+2d)$.
\end{proposition}

\begin{proof}
We show that the state of the triangle network can be decomposed into
four matrices, three of which are separable along one bipartition
each, and the fourth of which is fully separable. Using the same convention for subscripts as in the proof of Proposition \ref{thm:triangle-GME}, we notice that
\begin{equation}
\sigma_{\textnormal{cc}}(p)=\frac{p(p^{2}-3p+3)}{3}\left(\sigma_{1}+\sigma_{2}+\sigma_{3}\right)+(1-p)^{3}\tilde{\one}\otimes\tilde{\one}\otimes\tilde{\one},
\end{equation}
where
\begin{equation}
\begin{aligned} \sigma_{1}&= \frac{p^{2}\phi^{+}\otimes\phi^{+}\otimes\phi^{+}+\left(3p(1-p)/2\right)(\phi^{+}\otimes\tilde{\one}\otimes\phi^{+}+\tilde{\one}\otimes\phi^{+}\otimes\phi^{+})+3(1-p)^{2}\tilde{\one}\otimes\tilde{\one}\otimes\phi^{+}}{p^{2}-3p+3}\\
\sigma_{2}&=\frac{p^{2}\phi^{+}\otimes\phi^{+}\otimes\phi^{+}+\left(3p(1-p)/2\right)(\phi^{+}\otimes\phi^{+}\otimes\tilde{\one}+\tilde{\one}\otimes\phi^{+}\otimes\phi^{+})+3(1-p)^{2}\tilde{\one}\otimes\phi^{+}\otimes\tilde{\one}}{p^{2}-3p+3}\\
\sigma_{3}&=\frac{p^{2}\phi^{+}\otimes\phi^{+}\otimes\phi^{+}+\left(3p(1-p)/2\right)(\phi^{+}\otimes\tilde{\one}\otimes\phi^{+}+\phi^{+}\otimes\phi^{+}\otimes\tilde{\one})+3(1-p)^{2}\phi^{+}\otimes\tilde{\one}\otimes\tilde{\one}}{p^{2}-3p+3}.
\end{aligned}
\end{equation}
Clearly, the matrix $\tilde{\one}\otimes\tilde{\one}\otimes\tilde{\one}$
is fully separable. We will show that $\sigma_{1}$ is separable in
$A|BC$, and separability of $\sigma_{2}$ and $\sigma_{3}$ in $B|AC$
and $C|AB$ respectively will follow by symmetry. We have that
\begin{equation}
\sigma_{1}=\tau_{1}\otimes\phi_{B_{2}C_{2}}^{+},
\end{equation}
where showing that
\begin{equation}
\tau_{1}=\frac{p^{2}\phi_{A_{1}B_{1}}^{+}\otimes\phi_{A_{2}C_{1}}^{+}+\left(3p(1-p)/2\right)(\phi_{A_{1}B_{1}}^{+}\otimes\tilde{\one}_{A_{2}C_{1}}+\tilde{\one}_{A_{1}B_{1}}\otimes\phi_{A_{2}C_{1}}^{+})+3(1-p)^{2}\tilde{\one}_{A_{1}B_{1}}\otimes\tilde{\one}_{A_{2}C_{1}}}{p^{2}-3p+3}
\end{equation}
is separable in $A_{1}A_{2}|B_{1}C_{1}$ is sufficient to show that
$\sigma_{1}$ is separable in $A|BC$. Indeed, we can write
\begin{equation}
\tau_{1}=\frac{(3-p)^{2}}{4(3-3p+p^{2})}\left(\frac{2p}{3-p}\phi_{A_{1}B_{1}}^{+}+\frac{3(1-p)}{3-p}\tilde{\one}_{A_{1}B_{1}}\right)\otimes\left(\frac{2p}{3-p}\phi_{A_{2}C_{1}}^{+}+\frac{3(1-p)}{3-p}\tilde{\one}_{A_{2}C_{1}}\right)+\frac{3(1-p)^{2}}{4(3-3p+p^{2})}\tilde{\one}_{A_{1}B_{1}}\otimes\tilde{\one}_{A_{2}C_{1}}.
\label{eq:tauisotropic}
\end{equation}
The isotropic state
\begin{equation}
\frac{2p}{3-p}\phi^{+}+\frac{3(1-p)}{3-p}\tilde{\one}
\end{equation}
is separable whenever $2p/(3-p)\leq1/(d+1),$ i.e., $p\leq3/(3+2d),$
therefore $\tau_{1}$ is fully separable in $A_{1}|A_{2}|B_{1}|C_{1}$
which guarantees the required separability of $\sigma_{1}.$ Reasoning
symmetrically, separability of $\sigma_{2}$ in $B|AC$ and of $\sigma_{3}$
in $C|AB$ follows for the same values of $p$, and, hence, $\sigma_{\textnormal{cc}}(p)$
is biseparable for the stated bounds.
\end{proof}

\subsection{GME in tree networks of arbitrary size is not robust}

In this section we prove Theorem 1 in the main text, which we restate here.

\bigskip
\noindent \textbf{Theorem 1. }\emph{
Let $G=(V,E)$ be a tree graph with $n$ vertices and let $\sigma_{G}(p)$ denote the corresponding $n$-partite isotropic PEN state as given by Eq.\ (2) in the main text. Then, $\sigma_{G}(p)$ is biseparable if $|E|\geq dp/(1-p)$.}

\begin{proof}
Expanding the tensor product, we find
\begin{equation}
\sigma_{G}(p)=p^{|E|}\bigotimes_{e\in E}\phi_{d,e}^{+}+p^{|E|-1}(1-p)\sum_{e\in E}\tilde{\uno}_{e}\otimes\bigotimes_{e'\neq e}\phi_{d,e'}^{+}+\dots\,,\label{eq:treeBS}
\end{equation}
where the omitted terms are all separable along at least one bipartition,
since each bipartition is crossed by exactly one edge and at least
one edge in each term contains $\tilde{\one}.$ Showing that
the above expression is biseparable for $|E|$ as in the statement is sufficient to
prove the claim. But we can rewrite the above as
\begin{equation}
\begin{aligned}
\sigma_{G}(p)=\sum_{e\in E}\left(\frac{p^{|E|}}{|E|}\phi_{d,e}^{+}+p^{|E|-1}(1-p)\tilde{\uno}_{e}\right)\otimes\bigotimes_{e'\neq e}\phi_{d,e}^{+}+\dots\,.
\end{aligned}
\end{equation}
Here, each bracket has $\phi_{d}^{+}$ on $|E|-1$ edges, and the state
\begin{equation}
\frac{(p^{|E|}/|E|)\phi_{d}^{+}+p^{(|E|-1)}(1-p)\tilde{\one}}{p^{|E|}/|E|+p^{(|E|-1)}(1-p)}=\frac{(p/|E|)\phi_{d}^{+}+(1-p)\tilde{\one}}{p/|E|+1-p}\label{eq:treeBS-seppart}
\end{equation}
on the rest. But this is an isotropic state with visibility $(p/|E|)/(p/|E|+1-p)$,
which is thus guaranteed to become separable when the visibility is
smaller than or equal to $1/(d+1)$. For any fixed $p<1$, this can be achieved
by choosing $|E|\geq dp/(1-p).$
\end{proof}

Notice that the bound of Theorem \ref{thm1} is not optimal, as lower $|E|$ could be achieved by distributing the term $\bigotimes_{e\in E}\phi_{d,e}^{+}$ among some or all of the omitted terms in equation (\ref{eq:treeBS}) as well, as in the proof of Proposition \ref{thm:lambda-iso-BS}.

It is worth pointing out that the technique used in the proof of Theorem \ref{thm1} can be used to establish an analogous result for a much more general class of noise models, namely, when the states at the edges are convex mixtures of an entangled state and a separable
state that is not on the boundary of the set of separable states. More interestingly, these ideas can be extended to show that the non-robustness of GME holds not only for tree networks, but also for graphs that contain cycles. We close this section by proving this for the $n$-cycle, which is a graph of $n$ vertices containing a single cycle through all vertices, and we denote the corresponding isotropic PEN state by $\sigma_{\textnormal{cycle}}(p)$ as given by Eq.\ (2) in the main text.

\begin{theorem}\label{thm:polygon-BS}
For any fixed $d$ and $p<1$, there exists $n\in\mathbb{N}$
such that the $n$-partite $d$-dimensional isotropic PEN state $\sigma_\textnormal{cycle}(p)$ is biseparable.
\end{theorem}

\begin{proof}
Notice that, setting the convention $n+1\equiv1$, for this network $E=\{(i,i+1):i\in V\}$ and $|E|=|V|=n$. Thus, in a slight abuse of notation we will denote a state $\rho_e$ corresponding to $e=(i,i+1)$ by $\rho_i$. Expanding the
tensor product, we find
\begin{equation}
\sigma_{\textnormal{cycle}}(p)=p^{n}\bigotimes_{i=1}^{n}\phi_{d,i}^{+}+p^{n-1}(1-p)\sum_{i=1}^{n}\tilde{\one}_{i}\otimes\bigotimes_{j\neq i}\phi_{d,j}^{+}+p^{n-2}(1-p)^{2}\sum_{\substack{i,j=1\\
i\neq j
}
}^{n}\tilde{\one}_{i}\otimes\tilde{\one}_{j}\otimes\bigotimes_{k\neq i,j}\phi_{d,k}^{+}+\dots\,,
\label{eq:polygonBS}
\end{equation}
where all terms with two or more edges containing $\tilde{\one}$
are separable along at least one bipartition, since each bipartition
of the cycle is crossed by exactly two edges. We will show that
the terms where fewer than two edges contain $\tilde{\one}$
can be paired with separable terms in order to write $\sigma_\textnormal{cycle}(p)$
as a convex mixture of biseparable states. Out of the terms containing
$\tilde{\one}$ on two edges, there are $n$ such terms containing
$\tilde{\one}_{i}\otimes\tilde{\one}_{i+1}$ for some $i$. This means that the term is separable
in $i+1|\overline{i+1}$. Pairing these $n$ terms with
$p^{n}\bigotimes_{i=1}^{n}\phi_{d,i}^{+}$, we can write a fragment
of $\sigma_\textnormal{cycle}(p)$ as
\begin{equation}
\begin{aligned} & p^{n}\bigotimes_{i=1}^{n}\phi_{d,i}^{+}+p^{n-2}(1-p)^{2}\left(\sum_{i=1}^{n}\tilde{\one}_{i}\otimes\tilde{\one}_{i+1}\otimes\bigotimes_{j\neq i,i+1}\phi_{d,j}^{+}\right)\\
 & =p^{n-2}\sum_{i=1}^{n}\left[\left(\frac{p^{2}}{n}\phi_{d,i}^{+}\otimes\phi_{d,i+1}^{+}+(1-p)^{2}\tilde{\one}_{i}\otimes\tilde{\one}_{i+1}\right)\otimes\bigotimes_{j\neq i,i+1}\phi_{d,j}^{+}\right].
\end{aligned}
\label{eq:fragment1}
\end{equation}
If, for each $i$, the state
\begin{equation}
\frac{p^{2}}{n}\phi_{d,i}^{+}\otimes\phi_{d,i+1}^{+}+(1-p)^{2}\tilde{\one}_{i}\otimes\tilde{\one}_{i+1},\label{eq:iso2}
\end{equation}
once normalized, is separable in $i+1|\overline{i+1}$,
then the fragment of $\sigma_\textnormal{cycle}(p)$ in equation (\ref{eq:fragment1})
will be biseparable. Now, the state in Eq.\ (\ref{eq:iso2}) is a convex mixture
of $\phi_{d,i}^{+}\otimes\phi_{d,i+1}^{+}$ and $\tilde{\one}_{i}\otimes\tilde{\one}_{i+1}$.
Therefore, the state (\ref{eq:iso2}) is guaranteed to become separable
when $(1-p)^{2}/(p^{2}/n+(1-p)^{2})$ is close enough to
1. For fixed $p<1$, this can be achieved by choosing a large enough $n$.

Using a similar strategy, the terms containing $\tilde{\one}$
on one edge can be combined with some of those containing $\tilde{\one}$
on three appropriately chosen edges. Thus, assuming that $n>4$, another
fragment of $\sigma_\textnormal{cycle}(p)$ can be written as
\begin{equation}
\begin{aligned} & p^{n-1}(1-p)\tilde{\one}_{i}\otimes\bigotimes_{j\ne i}\phi_{d,j}^{+}+p^{n-3}(1-p)^{3}\sum_{\substack{j\neq i,i\pm1,\\
i-2
}
}\tilde{\one}_{i}\otimes\tilde{\one}_{j}\otimes\tilde{\one}_{j+1}\otimes\bigotimes_{k\neq i,j,j+1}\phi_{d,k}^{+}\\
 & =p^{n-3}(1-p)\sum_{\substack{j\neq i,i\pm1,\\
i-2
}
}\left[\tilde{\one}_{i}\otimes\bigotimes_{k\neq i,j,j+1}\phi_{d,k}^{+}\otimes\left(\frac{p^{2}}{n-4}\phi_{d,j}^{+}\otimes\phi_{d,j+1}^{+}+(1-p)^{2}\tilde{\one}_{j}\otimes\tilde{\one}_{j+1}\right)\right].
\end{aligned}
\label{eq:fragment2}
\end{equation}
Hence, it is sufficient to show that the state
\begin{equation}
\frac{p^{2}}{n-4}\phi_{d,j}^{+}\otimes\phi_{d,j+1}^{+}+(1-p)^{2}\tilde{\one}_{j}\otimes\tilde{\one}_{j+1},
\end{equation}
once normalized, is separable in $j+1|\overline{j+1}$
to deduce that the fragment in equation (\ref{eq:fragment2}) is biseparable.
Again, for fixed $p$, this is guaranteed for large enough $n$. Since
every term that does not appear in fragments (\ref{eq:fragment1})
and (\ref{eq:fragment2}) is already biseparable, the claim follows.
\end{proof}

\section{GME in the completely connected network of arbitrary size is robust}

In this section we prove Theorem 2 in the main text, which we state again here.

\bigskip
\noindent \textbf{Theorem 2. }\emph{
Let $G$ be a completely connected graph of $n$ vertices and let $\sigma_{\textnormal{cc}}(p)$ denote the corresponding $n$-partite isotropic PEN state as given by Eq.\ (2) in the main text. Then, there exists a value of $p_0<1$, which is independent of $n$ (i.e.\ depends only on $d$), such that $\sigma_{\textnormal{cc}}(p)$ is GME for every $n$ and for all $p>p_0$.}

The claim follows from the next two lemmas. Therein, we consider the fidelity between two quantum states $\rho$ and $\sigma$
\begin{equation}\label{fidelity}
F(\rho,\sigma)=\tr^2\sqrt{\sqrt{\rho}\sigma\sqrt{\rho}}
\end{equation}
which, when one of the states is pure, boils down to
\begin{equation}\label{fidelitypure}
F(\rho,|\psi\rangle\langle\psi|)=\tr(\rho|\psi\rangle\langle\psi|).
\end{equation}
The fidelity is often defined to be the square root of the expression given in Eq.\ (\ref{fidelity}). Our choice is motivated by the fact that the fidelity in Eq.\ (\ref{fidelitypure}) is linear in both arguments.

\begin{lemma}\label{lem:fidelity-bs}
Let $\phi_{ij}^{+}$ be the 2-dimensional maximally entangled state shared by parties $i$ and $j$ and let
\begin{equation}
\Lambda_{ij}:\bigotimes_{i=1}^{n}\mathcal{B}(\mathcal{H}_{i})\rightarrow\mathcal{B}(\mathcal{H}_{i})\otimes\mathcal{B}(\mathcal{H}_{j})
\end{equation}
be an arbitrary LOCC protocol that maps $n$-partite states to bipartite states shared between parties $i$ and $j$. Then, for every $n$-partite biseparable state $\chi$ it holds that
\begin{equation}\label{eq:fidelityBSlemma}
\frac{2}{n(n-1)}\sum_{i<j}F(\Lambda_{ij}(\chi),\phi_{ij}^{+})\leq1-\frac{1}{n}.
\end{equation}
\end{lemma}
\begin{proof}
Since $\chi$ is biseparable,
we have
\begin{equation}
\chi=\sum_{M}p_{M}\chi_{M},
\end{equation}
where each $\chi_{M}$ is separable across the bipartition $M|\overline{M}$
and $\sum_{M}p_{M}=1$. Let $M_{ij}$ be the set of bipartitions that
split parties $i$ and $j$. Then, we can write the evolution
of $\chi$ under the protocol $\Lambda_{ij}$ as
\begin{equation}
\Lambda_{ij}(\chi)=\sum_{M\not\in M_{ij}}p_{M}\tau_{M}(i,j)+\sum_{M\in M_{ij}}p_{M}\sigma_{M}(i,j),
\end{equation}
where $\tau_{M}(i,j)$ and $\sigma_{M}(i,j)$ are bipartite states
of parties $i$ and $j$. The state $\tau_{M}(i,j)$ is in principle
unrestricted, so it can have up to unit fidelity with $\phi_{ij}^{+}$. However, since LOCC operations
cannot create entanglement, $\sigma_{M}(i,j)$ must be separable,
therefore its fidelity with $\phi_{ij}^{+}$ cannot be larger than
$1/2$. Therefore, for any LOCC protocol $\Lambda_{ij}$, it holds that
\begin{equation}
F(\Lambda_{ij}(\chi),\phi_{ij}^{+})\leq\sum_{M\not\in M_{ij}}p_{M}+\frac{1}{2}\sum_{M\in M_{ij}}p_{M}=1-\frac{1}{2}\sum_{M\in M_{ij}}p_{M}.
\end{equation}
Finally, summing over all parties $i<j$, we obtain
\begin{equation}\label{eq:fidelityBS}
\sum_{i<j}F(\Lambda_{ij}(\chi),\phi_{ij}^{+})\leq\frac{n(n-1)}{2}-\frac{n-1}{2},
\end{equation}
which readily gives Eq.\ (\ref{eq:fidelityBSlemma}). To see this, notice that the first term in Eq.\ (\ref{eq:fidelityBS}) is the number of distinct pairs $i<j$, while the second comes from the following observation: the sums over $i<j$ and $M\in M_{ij}$
run through all bipartitions $M$, more than once. In fact, the number
of times each bipartition is counted is equal to the number of times
each bipartition is crossed by the edge connecting $i$ and $j$.
In turn, this number is equal to the number of edges crossing each
bipartition. A bipartition splitting $k$ parties from the remaining
$(n-k)$ is crossed by $k(n-k)$ edges, which is smallest when $k=1$.
That is, each bipartition $M\in M_{ij}$ appears at least $(n-1)$
times, and so the sum $\sum_{M}p_{M}$, running over \emph{all }$M$,
appears at least $(n-1)$ times. That is,
\begin{equation}
\sum_{i<j}\sum_{M\in M_{ij}}p_{M}\geq n-1.
\end{equation}
\end{proof}

\begin{lemma}\label{lem2:fidelity-bs}
With the same notation as in Lemma \ref{lem:fidelity-bs}, for any fixed $d$, there is a fixed value of $\tilde p<1$ and particular LOCC protocols $\{\Lambda_{ij}\}$ acting on the $n$-partite $d$-dimensional isotropic PEN state $\sigma_{\textnormal{cc}}(p)$ such that
\begin{equation}\label{eq:fidelityGMElemma}
\frac{2}{n(n-1)}\sum_{i<j}F(\Lambda_{ij}(\sigma_{\textnormal{cc}}(p)),\phi_{ij}^{+})>1-\frac{1}{n}.
\end{equation}
for all $p>\tilde p$ if $n$ is sufficiently large.
\end{lemma}

\begin{proof}
The protocol $\Lambda_{ij}$ for each fixed pair $i<j$ goes as follows. In the complete graph, each party $k\neq i,j$ shares isotropic states with these parties: $\rho_{ik}(p)$ and $\rho_{jk}(p)$. The protocol starts with each $k$ ($k\neq i,j$) teleporting their half of $\rho_{ik}(p)$ to $j$ by using $\rho_{jk}(p)$ as a channel. This is a noisy
version of the standard teleportation protocol \cite{bennett_teleporting_1993}. The teleported state will be a mixture of four terms, namely the four combinations of teleporting half of a maximally entangled or maximally
mixed state along a maximally entangled or a maximally mixed channel.
The first term, with weight $p^{2}$, will give a maximally entangled
state, the other three turn out to give a maximally mixed state (as
can be simply checked by performing the calculations in the standard
protocol, replacing the teleported state and/or the channel by an
identity in each case). Therefore, the teleportation protocol yields
\begin{equation}
\rho_{ij}(p^{2})=p^{2}\phi_{d}^{+}+(1-p^{2})\tilde{\one}.
\end{equation}
Thus, parties $i$ and $j$ end up sharing $(n-2)$ perfect copies of
this state, one coming from each party $k\neq i,j$. If $p$ is sufficiently large, $i$ and $j$ can now apply a distillation protocol $D_{ij}$ to
obtain something close to a maximally entangled state, whose fidelity
approaches 1 in the limit of large $n$. More specifically,
\begin{equation}
F(\Lambda_{ij}(\sigma_{\textnormal{cc}}(p)),\phi_{ij}^{+})=F(D_{ij}((\rho_{ij}(p^{2}))^{\otimes(n-2)}),\phi_{ij}^{+})\geq1-\varepsilon_{n},\label{eq:fidelityGME}
\end{equation}
where $\varepsilon_{n}\rightarrow0$ as $n\rightarrow\infty$. To
compare to equation (\ref{eq:fidelityGMElemma}),
there only remains to show that $\varepsilon_{n}\rightarrow0$ sufficiently
fast as $n$ grows.

Reference \cite{bennett_mixed-state_1996} showed that a one-way distillation
protocol acting on isotropic states and having rate $R$ is equivalent
to a quantum error-correcting code on a depolarizing channel with
the same rate $R$. In turn, Ref. \cite{hamada_exponential_2002}
proved a lower bound on the fidelity $F$ of a $d$-dimensional quantum
error-correcting code of rate $R$ acting on a certain class of memoryless
channels, which, in particular, include depolarizing channels. After
$n$ uses of the channel, the lower bound \footnote{There is an extra factor of 2 here as compared to the expression in \cite{hamada_exponential_2002} due to the different definitions of fidelity.} is
\begin{equation}
F\geq1-2(n+1)^{2(d^{2}-1)}d^{-nE}.
\end{equation}
Here, $E$ is a function of the rate $R$ and the noise parameter
of the depolarizing channel, which, in turn, corresponds to a function
of the noise parameter $p$ of the isotropic state. It holds that
$E>0$ when the rate is strictly below the maximum achievable rate
of the channel \cite{hamada_exponential_2002}. By the correspondence with distillation, this entails
that $\varepsilon_{n}$ in equation (\ref{eq:fidelityGME}) goes to
zero exponentially fast for one-way distillable isotropic states if
the rate is suboptimal. Since, for our protocol $\Lambda_{ij}$, we
are only interested in obtaining one copy of $\phi^{+}$, we can achieve
this exponential decay. Therefore,
\begin{equation}
\frac{2}{n(n-1)}\sum_{i<j}F(\Lambda_{ij}(\sigma_{\textnormal{cc}}(p)),\phi_{ij}^{+})\geq1-\varepsilon_{n},
\end{equation}
with $\varepsilon_n<1/n$ if $n$ is large enough.

There remains to bound the visibility $p$ of the isotropic states
for which the statement holds if $n$ is large enough. The one-way
distillation protocol \cite{1waydistillation} requires that
the states $\eta$ to be distilled are such that
\begin{equation}
H(\tr_{A}(\eta))-H(\eta)>0,\label{eq:entropy-eta}
\end{equation}
where $H(\cdot)$ is the von Neumann entropy. It can be readily seen that
isotropic states $\eta=\rho(p^{2})$ of any fixed dimension satisfy this
inequality if $p$ is large enough but strictly smaller than one.
\end{proof}

Notice that in the particular case of $d=2$, Eq.\ (\ref{eq:entropy-eta}) reduces to
\begin{equation}
\frac{3(1-p^{2})}{4}\log_{2}(1-p^{2})+\frac{1+3p^{2}}{4}\log_{2}(1+3p^{2})>1,
\end{equation}
which holds when $p\gtrsim0.865$, which is the estimate mentioned in the main text.

Lemmas \ref{lem:fidelity-bs} and \ref{lem2:fidelity-bs} immediately imply the statement of Theorem \ref{thm2} under the extra assumption that $n$ is large enough. However, this readily gives that the claim must hold for all $n$. Indeed, for a fixed $d$, let $n_{0}\in\mathbb{N}$ be such that $\sigma_{\textnormal{cc}}(p)$ is GME for all $n>n_{0}$ if $p>\tilde{p}$. Now, for every fixed $n\leq n_0$, $\sigma_{\textnormal{cc}}(1)$ is GME and, since the set of GME states is open,
the state must remain GME for all $p>p^{*}(n)$, where $p^{*}(n)<1$.
Thus, the statement of the Theorem holds by picking
\begin{equation}
p_{0}=\max\{\tilde{p},\max_{n\leq n_{0}}p^{*}(n)\}.
\end{equation}
In fact, $p_{0}$ can be estimated as a function of $n_{0}$. Consider
$\sigma_{\textnormal{cc}}(p)$ with 2-dimensional isotropic states and the LOCC
protocol $\Lambda_{ij}$ that traces out all particles
except those corresponding to the isotropic state $\rho_{ij}(p)$. Then, the fidelity with the maximally entangled
state is
\begin{equation}
F(\Lambda_{ij}(\sigma_{\textnormal{cc}}(p)),\phi_{ij}^{+})=\frac{1+3p}{4}
\end{equation}
for all $i<j$. Comparing to the biseparability bound in Lemma
\ref{lem:fidelity-bs}, the $n$-partite state $\sigma_{\textnormal{cc}}(p)$ is GME
if $p>1-4/(3n)$ and, hence, $\max_{n\leq n_{0}}p^{*}(n)\leq1-4/(3n_{0})$.

\subsection{Construction of GME but non-GMNL states that display superactivation of GMNL}

In this section we prove Theorem 3 in the main text. For the sake of completeness we state it again.

\bigskip
\noindent \textbf{Theorem 3. }\emph{
Let $\tau(p)$ denote the $n$-partite PEN state corresponding to a star graph in which all edges represent the maximally entangled state except one, which is given by the isotropic state $\rho(p)$. Then, if}
\begin{equation}\label{nonsteerable}
\frac{1}{d+1}<p\leq\frac{(3d-1)(d-1)^{d-1}}{(d+1)d^d},
\end{equation}
\begin{itemize}
\item[(i)] \emph{$\tau(p)$ is GME $\forall n\geq3$.}
\item[(ii)] \emph{$\tau(p)$ is not GMNL $\forall n\geq3$.}
\item[(iii)] \emph{$\tau(p)^{\otimes k}$ is GMNL $\forall n\geq3$ if $k$ is large enough.}
\end{itemize}

The first two items in the theorem are a consequence of the following two lemmas together with the fact that for the range of visibilities given in Eq.\ (\ref{nonsteerable}) a $d$-dimensional isotropic state is entangled but unsteerable \cite{almeida}. We recall that a bipartite state $\rho_{AB}$ has a local hidden-state (LHS) model if, for any POVMs on $A$'s side $\{M_{a|x}\}$ (where $x$ labels the POVM and $a$ the possible outcomes of each of them), it holds that
\begin{equation}\label{lhs}
\tr_A(M_{a|x}\otimes\one\rho_{AB})=\sum_\lambda p(\lambda)P_A(a|x\lambda)\sigma_\lambda,
\end{equation}
where $p(\lambda)$ is a probability distribution, $P_A(a|x\lambda)$ is a conditional probability distribution given $x$ and $\lambda$, and $\sigma_\lambda$ is a single-party state which depends on the hidden variable $\lambda$. If a state $\rho_{AB}$ does not have an LHS model of the form of equation (\ref{lhs}) for some measurements, then $\rho_{AB}$ is said to be steerable from $A$ to $B$. Exchanging the roles of $A$ and $B$ allows one to define steerability from $B$ to $A$. However, the fact that the isotropic state is permutationally invariant renders this distinction irrelevant in our case.

\begin{lemma}
The $n$-partite PEN state over a star network (with the central party labeled by $A=A_1\cdots A_{n-1}$ and the other parties by $\{B_i\}_{i=1}^{n-1}$),
\begin{equation}
\tau(p)=\rho_{A_1B_1}(p)\otimes\bigotimes_{i=2}^{n-1}\phi^+_{A_iB_i},
\end{equation}
is GME if the isotropic state $\rho_{A_1B_1}(p)$ is entangled, i.e.\ if $p>1/(d+1)$.
\end{lemma}
\begin{proof}
We will show that, if $A$ applies
\begin{equation}
X=\sum_{i=0}^{d-1}\ket{i}\bra{i}^{\otimes n-1},
\end{equation}
and the $\{B_i\}$ apply the identity, the resulting state is GME if $p>1/(d+1)$. Since
local operators preserve biseparability, this will mean that $\tau(p)$
is GME too. Let
\begin{equation}
\tilde{\tau}_{f}(p)=\left(X_{A}\otimes\one_{B_{1}...B_{n-1}}\right)\tau\left(X_{A}^{\dagger}\otimes\one_{B_{1}...B_{n-1}}\right)
\end{equation}
be the unnormalized state after the parties apply their operations.
Since
\begin{equation}
\tau(p)=p\phi_{A_1B_{1}}^{+}\otimes\phi_{A_2B_{2}}^{+}\otimes\dots\otimes\phi_{A_{n-1}B_{n-1}}^{+}+(1-p)\tilde{\one}_{A_1B_{1}}\otimes\phi_{A_2B_{2}}^{+}\otimes\dots\otimes\phi_{A_{n-1}B_{n-1}}^{+}.
\end{equation}
and
\begin{equation}
\begin{aligned}\phi_{AB_{1}}^{+}\otimes\phi_{AB_{2}}^{+}\otimes\dots\otimes\phi_{AB_{n-1}}^{+} & =\frac{1}{d^{n-1}}\sum_{i,j=0}^{d^{n-1}-1}\ket{i}_{A}\ket{i}_{B_{1}...B_{n-1}}\bra{j}_{A}\bra{j}_{B_{1}...B_{n-1}}\\
\tilde{\one}_{AB_{1}}\otimes\phi_{AB_{2}}^{+}\otimes\dots\otimes\phi_{AB_{n-1}}^{+} & =\frac{1}{d^{n}}\sum_{i,j=0}^{d-1}\sum_{k,\ell=0}^{d^{n-2}-1}\ket{ik}_{A}\ket{jk}_{B_{1}...B_{n-1}}\bra{i\ell}_{A}\bra{j\ell}_{B_{1}...B_{n-1}}.
\end{aligned}
\end{equation}
Applying $X$ to each $\ket{i}_{A}$, where $i=0,...,d^{n-1}-1$,
picks out the terms where all digits of $i$ are equal, and thus gives
simply $\ket{i}_{A}$ with $i=0,...,d-1$. Then, the digits of $\ket{i}_{B_{1}...B_{n-1}}$
must also be equal. Similarly, the action of $X$ on $\ket{ik}_{A}$,
where $i=0,...,d-1$ and $k=0,...,d^{n-2}-1$, makes $i=k=0,...,d-1$,
with the corresponding effect on the $k$ index of $\ket{jk}_{B_{1}...B_{n-1}}$.
Therefore, we obtain
\begin{equation}
\tilde{\tau}_{f}(p)=\frac{p}{d^{n-1}}\sum_{i,j=0}^{d-1}\ket{i}_{A}\ket{i}_{B_{1}...B_{n-1}}^{\otimes n-1}\bra{j}_{A}\bra{j}_{B_{1}...B_{n-1}}^{\otimes n-1}+\frac{1-p}{d^{n}}\sum_{i,j=0}^{d-1}\ket{i}_{A}\ket{j}_{B_{1}}\ket{i}_{B_{2}...B_{n}}^{\otimes n-2}\bra{i}_{A}\bra{j}_{B_{1}}\bra{i}_{B_{2}...B_{n}}^{\otimes n-2}.
\end{equation}
Hence, the normalized state after the transformation is
\begin{equation}
\tau_{f}(p)=\frac{\tilde{\tau}_{f}}{\tr\tilde{\tau}_{f}}=p GHZ_{AB_1\cdots B_{n-1}}+\frac{1-p}{d^{2}}\sum_{i,j=0}^{d-1}\ket{i}_{A}\ket{j}_{B_{1}}\ket{i}_{B_{2}...B_{n}}^{\otimes n-2}\bra{i}_{A}\bra{j}_{B_{1}}\bra{i}_{B_{2}...B_{n}}^{\otimes n-2},
\end{equation}
where $GHZ=|GHZ\rangle\langle GHZ|$ and $|GHZ\rangle=\sum_{i=0}^{d-1}|i\rangle^{\otimes n}/\sqrt{d}$ is the $n$-partite $d$-dimensional GHZ state. To show that this state is GME, it is sufficient to find a witness that detects it. The operator
\begin{equation}
W=\frac{1}{d}\one_{AB_1\cdots B_{n-1}}-GHZ_{AB_1\cdots B_{n-1}}
\end{equation}
fits the bill: since the maximum overlap of the GHZ state with a biseparable
state is $1/d$ \cite{biswas_genuine-multipartite-entanglement_2014}, we have that
\begin{equation}
\tr(W\sigma)\geq0
\end{equation}
for all biseparable states $\sigma.$ Moreover,
\begin{equation}
\tr(W\tau_{f}(p))=\frac{d-1-p(d^{2}-1)}{d^{2}},
\end{equation}
which is strictly smaller than zero when
\begin{equation}
p>\frac{1}{d+1}.
\end{equation}
\end{proof}

\begin{lemma}\label{lemma steerable}
Any $n$-partite PEN state over a star network (with the central party labeled by $A=A_1\cdots A_{n-1}$ and the other parties by $\{B_i\}_{i=1}^{n-1}$),
\begin{equation}
\bigotimes_{i=1}^{n-1}\rho^{(i)}_{A_iB_i},
\end{equation}
is not GMNL if at least one of the bipartite states $\rho^{(i)}_{A_iB_i}$ is nonsteerable from $B_i$ to $A_i$.
\end{lemma}

\begin{proof}
We can assume that $\rho^{(1)}_{A_1B_1}$ is nonsteerable from $B_1$ to $A_1$ without loss of generality. We will show that any probability distribution arising from local POVMs on the network state is local across the bipartition $B_1|AB_2\cdots B_{n-1}$. In particular, this shows that the state is bilocal.

Let us assume that party $A$ measures according to the POVM $\{ E_{a|x}\} _{a}$ with outputs $a$ and inputs $x$ and, for every $i=1,\cdots n-1$, party $B_i$ measures according to the POVM $\{ F_{b_i|y_{i}}^{i}\} _{b_{i}}$ with outputs $b_i$ and inputs $y_i$. Then, they generate a probability distribution of the form
\begin{equation}
P(ab_{1}...b_{n-1}|xy_{1}...y_{n-1})=\tr\left[\left(E_{a|x}\otimes \bigotimes_{i=1}^{n-1}F_{b_i|y_{i}}^{i}\right)\left(\rho^{(1)}_{A_1B_1}\otimes \bigotimes_{i=2}^{n-1}\rho^{(i)}_{A_iB_i}\right)\right].
\end{equation}
This distribution can be rewritten as
\begin{equation}
\tr_{A_1B_1}\left(\left(\tr_{A_2\cdots A_{n-1}B_2\cdots B_{n-1}}\left[\left(E_{a|x}\otimes \bigotimes_{i=2}^{n-1}F_{b_i|y_{i}}^{i}\right)\left(\one_{A_1}\otimes \bigotimes_{i=2}^{n-1}\rho^{(i)}_{A_iB_i} \right)\right]\otimes F_{b_1|y_{1}}^{1}\right)\rho^{(1)}_{A_1B_1}\right),
\label{eq:bilocal-sigma}
\end{equation}
where we will show that
\begin{equation}
\begin{aligned}F_{ab_2\cdots b_{n-1}|xy_2\cdots y_{n-1}} & :=\tr_{A_2\cdots A_{n-1}B_2\cdots B_{n-1}}\left[\left(E_{a|x}\otimes \bigotimes_{i=2}^{n-1}F_{b_i|y_{i}}^{i}\right)\left(\one_{A_1}\otimes \bigotimes_{i=2}^{n-1}\rho^{(i)}_{A_iB_i} \right)\right]\end{aligned}
\label{eq:bilocal-mtsH}
\end{equation}
is a POVM element acting on $\mathcal{H}_{A_1}$. By denoting $A_1=:\tilde{A}$, $A_2\cdots A_{n-1}B_2\cdots B_{n-1}=:\tilde{B}$, $ab_{2}...b_{n-1}=:\tilde{a}$, $xy_{2}...y_{n-1}=:\tilde{x}$, $E_{a|x}\otimes \bigotimes_{i=2}^{n-1}F_{b_i|y_{i}}^{i}=:E_{\tilde{a}|\tilde{x}}^{\tilde{A}\tilde{B}}$ and $\bigotimes_{i=2}^{n-1}\rho^{(i)}_{A_iB_i}=:\tau_{\tilde{B}} $, we can rewrite Eq. (\ref{eq:bilocal-mtsH}) as
\begin{equation}
F_{\tilde{a}|\tilde{x}}=\tr_{\tilde{B}}\left(E_{\tilde{a}|\tilde{x}}^{\tilde{A}\tilde{B}}(\one_{\tilde{A}}\otimes\tau_{\tilde{B}})\right).
\end{equation}

To show positivity, we notice that $\tau_{\tilde{B}}$ is a quantum state,
and thus can be written as a convex combination of pure states $\ket{\psi}\bra{\psi}$.
Therefore, to show positivity of $F_{\tilde{a}|\tilde{x}}$ we can assume
\begin{equation}
F_{\tilde{a}|\tilde{x}}=\tr_{\tilde{B}}\left(E_{\tilde{a}|\tilde{x}}^{\tilde{A}\tilde{B}}(\one_{\tilde{A}}\otimes\ket{\psi}\bra{\psi})\right)\equiv\bra{\psi}E_{\tilde{a}|\tilde{x}}^{\tilde{A}\tilde{B}}\ket{\psi}.
\end{equation}
If $F_{\tilde{a}|\tilde{x}}$ were not positive, there would exist $\ket{z}\in\mathcal{H}_{A}$
such that
\begin{equation}
\bra{z}F_{\tilde{a}|\tilde{x}}\ket{z}<0,
\end{equation}
which would imply that
\begin{equation}
\bra{\psi}\bra{z}E_{\tilde{a}|\tilde{x}}^{\tilde{A}\tilde{B}}\ket{z}\ket{\psi}<0,
\end{equation}
and hence that $E_{\tilde{a}|\tilde{z}}^{\tilde{A}\tilde{B}}$ would not be positive. But this is
false, as $E_{\tilde{a}|\tilde{x}}^{\tilde{A}\tilde{B}}$ is a POVM element. Therefore, $F_{\tilde{a}|\tilde{x}}\succcurlyeq0$.

Normalization of $F_{\tilde{a}|\tilde{x}}$ is guaranteed by the normalization of
$E_{\tilde{a}|\tilde{x}}^{\tilde{A}\tilde{B}}$, as
\begin{equation}
\begin{aligned}\sum_{\tilde{a}}F_{\tilde{a}|\tilde{x}} & =\tr_{\tilde{B}}\left(\sum_{\tilde{a}}E_{\tilde{a}|\tilde{x}}^{\tilde{A}\tilde{B}}(\one_{\tilde{A}}\otimes\tau_{\tilde{B}})\right)\\
 & =\tr_{\tilde{B}}\left(\one_{\tilde{A}}\otimes\tau_{\tilde{B}}\right)=\one_{\tilde{A}}.
\end{aligned}
\end{equation}

Therefore, since equation (\ref{eq:bilocal-sigma}) expresses the
distribution achieved by the parties in the network as two POVM elements
acting on the nonsteerable state $\rho^{(1)}_{A_1B_1}$ we can easily conclude that
\begin{equation}
P(ab_{1}...b_{n-1}|xy_{1}...y_{n-1})=\sum_{\lambda}p_{\lambda}P_{1}(b_{1}|y_{1}\lambda)P_{2}(ab_2\cdots b_{n-1}|xy_2\cdots y_{n-1}\lambda),
\end{equation}where
\begin{equation}
P_{2}(ab_2\cdots b_{n-1}|xy_2\cdots y_{n-1}\lambda)=tr(F_{ab_2\cdots b_{n-1}|xy_2\cdots y_{n-1}} \sigma_\lambda).
\end{equation}

Now, the particular form of $F_{ab_2\cdots b_{n-1}|xy_2\cdots y_{n-1}}$ defined in Eq. (\ref{eq:bilocal-mtsH}) ensures that $P_{2}$ is nonsignaling and thus $P(ab_{1}...b_{n-1}|xy_{1}...y_{n-1})$ is local local across the bipartition $B_1|AB_2\cdots B_{n-1}$.
\end{proof}

\begin{observation}Note that if the state $\rho^{(1)}_{A_1B_1}$ in Lemma \ref{lemma steerable} is local (but steerable) the previous proof still works, but we can only conclude that the distribution $P(ab_{1}...b_{n-1}|xy_{1}...y_{n-1})$ is bilocal according to Svetlichny's definition of bilocality \cite{svetlichny_distinguishing_1987}, where the probability distributions of the bipartition elements are not required to be nonsignaling. However, the operational definition that is used throughout this work requires, in addition, that the local components $P_{1}$ and $P_{2}$ are nonsignaling.
\end{observation}


In order to prove the third item in Theorem \ref{thm superactiv} we will show that taking many copies of $\tau(p)$ makes it possible to win a generalization of the Khot-Vishnoi game with a higher probability than with bilocal resources.

The Khot-Vishnoi game \cite{khot_unique_2005,buhrman_near-optimal_2012} is parametrized by a number $v$, which is assumed to be a power of 2, and a noise parameter $\eta\in[0,1/2]$. Consider the group $\{0,1\}^{v}$
of all $v$-bit strings, with operation $\oplus$ denoting bitwise
modulo 2 addition, and the subgroup $H$ of all Hadamard codewords.
The subgroup $H$ partitions the group $\{0,1\}^{v}$ into $2^{v}/v$
cosets of $v$ elements each. These cosets will act as questions,
and answers will be elements of the question cosets. The referee chooses
a uniformly random coset $[x]$, as well as a string $z\in\{0,1\}^{v}$
where each bit $z(i)$ is chosen independently and is 1 with probability
$\eta$ and 0 otherwise. Alice's question is the coset $[x]$, which
can be thought of as $u\oplus H$ for a uniformly random $u\in\{0,1\}^{v}$,
while Bob's is the coset $[x\oplus z]$, which can be thought of as
$u\oplus z\oplus H$. The aim of the players is to guess the string
$z$, and thus they must output $a\in[x]$ and $b\in[x\oplus z]$
such that $a\oplus b=z$. Ref. \cite{buhrman_near-optimal_2012} showed
that any local strategy for Alice and Bob, implemented by a distribution
denoted by $P_{\text{local}}$, achieves a winning probability of
\begin{equation}
\left\langle G_{KV},P_{\text{local}}\right\rangle \leq\frac{v}{v^{1/(1-\eta)}}.
\end{equation}

A higher winning probability can be obtained with a distribution
$P_{\text{max}}$ arising from certain projective measurements on
the maximally entangled state:
\begin{equation}
\left\langle G_{KV},P_{\text{max}}\right\rangle \geq(1-2\eta)^{2}.\label{eq:KV-quantum-bipart}
\end{equation}

Picking the value $\eta=1/2-1/\log v$ leads to the bounds
\begin{equation}
\left\langle G_{KV},P_{\text{local}}\right\rangle \leq\frac{C}{v}\quad\left\langle G_{KV},P_{\text{max}}\right\rangle \geq D/\log^{2}v,
\end{equation}
for universal constants $C,D$.

We will now introduce an extension of the previous game to the star network via a lemma. Then, we will prove a second lemma to bound the probability of winning with a bilocal strategy, before showing the superactivation result.
\begin{lemma}\label{KV general}
\label{lem:KVmulti}The Khot-Vishnoi game can be extended to the star
network by letting Alice and each of $Bob$$_{i}$, for $i=1,...,K$,
play the bipartite Khot-Vishnoi game. This defines a game whose coefficients
are normalized, i.e., satisfy
\begin{equation}
\sum_{\substack{x_{1},...,x_{K}\\
y_{1},...,y_{K}
}
}\max_{\substack{a_{1},...,a_{K}\\
b_{1},...,b_{K}
}
}\widetilde{G}_{a_{1}...a_{K}b_{1}...b_{K}|x_{1}...x_{K}y_{1}...y_{K}}\leq1.
\end{equation}
\end{lemma}
\begin{proof}
Consider a star network of $K$ edges, each of which connects Alice
to Bob$_{i}$ for $i=1,...,K$. Alice will play the bipartite Khot-Vishnoi
game $G_{KV}$ with each Bob$_{i}$. Denoting Alice's inputs and outputs
as $x_{1},...,x_{K}$ and $a_{1},...,a_{K}$ respectively, and Bob$_{i}$'s
input and output as $y_{i},$ $b_{i}$ respectively, we denote the
coefficients of each game as $G_{a_{i}b_{i}|x_{i}y_{i}}$. Hence, the
$(K+1)$-partite game being played on the star network, which we denote
by $\widetilde{G}$, has coefficients
\begin{equation}
\widetilde{G}_{a_{1}...a_{K}b_{1}...b_{K}|x_{1}...x_{K}y_{1}...y_{K}}=\prod_{i=1}^{K}G_{a_{i}b_{i}|x_{i}y_{i}}.
\end{equation}
Since $G_{KV}$ is a game, so is $\widetilde{G}$, i.e. all of its
coefficients are positive. Moreover, one can check that the coefficients
$G_{a_{i}b_{i}|x_{i}y_{i}}$ satisfy the normalization condition
\begin{equation}
\sum_{x_{i},y_{i}}\max_{a_{i},b_{i}}G_{a_{i}b_{i}|x_{i}y_{i}}\leq1\label{eq:KV-normn}
\end{equation}
for all $i\in[K]$, and hence
\begin{equation}
\begin{aligned}\sum_{\substack{x_{1},...,x_{K}\\
y_{1},...,y_{K}
}
}\max_{\substack{a_{1},...,a_{K}\\
b_{1},...,b_{K}
}
}\widetilde{G}_{a_{1}...a_{K}b_{1}...b_{K}|x_{1}...x_{K}y_{1}...y_{K}} & =\sum_{\substack{x_{1},...,x_{K}\\
y_{1},...,y_{K}
}
}\max_{\substack{a_{1},...,a_{K}\\
b_{1},...,b_{K}
}
}\prod_{i=1}^{K}G_{a_{i}b_{i}|x_{i}y_{i}}\\
 & =\prod_{i=1}^{K}\sum_{x_{i},y_{i}}\max_{a_{i},b_{i}}G_{a_{i}b_{i}|x_{i}y_{i}}\\
 & \leq1.
\end{aligned}
\label{eq:KVmulti-normn}
\end{equation}
In fact, this normalization condition also holds if we take only a
subset of games $G_{KV}$, i.e. the product of $G_{a_{i}b_{i}|x_{i}y_{i}}$
for $i$ in some subset of $[K]$.
\end{proof}
\begin{lemma}
\label{lem:KV-bl}The extension of the Khot-Vishnoi game to the star
network is such that the winning probability using any bilocal strategy
is bounded above by $C/v$, where $v$ is the parameter of the game
and $C$ is a universal constant.
\end{lemma}

\begin{proof}
Consider the game in Lemma \ref{lem:KVmulti}. In order to bound the winning
probability of bilocal strategies, we consider the two possible types
of bilocal distributions $P_{BL}$: those local in a bipartition that
separates Alice from all the Bobs, and those where some of the Bobs
are in Alice's partition element.

First, we take a distribution of the form
\begin{equation}
P_{1}(a_{1},...,a_{K}|x_{1},...,x_{K})P_{2}(b_{1},...,b_{K}|y_{1},...,y_{K})\label{eq:KVparallelprob}
\end{equation}
for each $a_{i},b_{i},x_{i},y_{i}$, $i\in[K]$.
Using such a distribution to play $\widetilde{G}$,
the parties are effectively playing the $K$-fold parallel repetition \cite{raz_parallel_1998} of $G_{KV}$, which we denote as $G_{KV}^{\otimes K}$. To bound their
winning probability we will use the same techniques as in Refs. \cite{amr_unbounded_2020, buhrman_near-optimal_2012}. Recall that, for the bipartite game, the questions are cosets of $H$
in the group $\{0,1\}^{v}$, which can be thought of as $u\oplus H$
for Alice and $u\oplus z\oplus H$ for Bob (where $u\in\{0,1\}^{v}$
is sampled uniformly and $z\in\{0,1\}^{v}$ is sampled bitwise independently
with noise $\eta$), and the answers are elements of the question
cosets. Without loss of generality, we can assume Alice and Bob's
strategy is deterministic, and identify it with Boolean functions
$A,B:\{0,1\}^{v}\rightarrow\{0,1\}$ which take the value 1 for exactly
one element of each coset. That is, for each question, $A,B$ respectively
pick out Alice's and Bob's answer. Since the players win if and only
if their answers $a,b$ satisfy $a\oplus b=z$, we have that for all
$u,z$,
\begin{equation}
\sum_{h\in H}A(u\oplus h)B(u\oplus z\oplus h)
\end{equation}
is 1 if the players win on inputs $u\oplus H$, $u\oplus z\oplus H$,
and 0 otherwise. Therefore, the winning probability is
\begin{equation}
\begin{aligned}\mathop{\mathbb{E}}_{u,z}\left[\sum_{h\in H}A(u\oplus h)B(u\oplus z\oplus h)\right] & =\sum_{h\in H}\mathop{\mathbb{E}}_{u,z}\left[A(u\oplus h)B(u\oplus z\oplus h)\right]\\
 & =v\mathop{\mathbb{E}}_{u,z}\left[A(u)B(u\oplus z)\right],
\end{aligned}
\end{equation}
since for all $h,$ the distribution of $u\oplus h$ is uniform.

For the parallel repetition $G_{KV}^{\otimes K}$, Alice and Bob must
pick an answer for each copy of the game, so we can identify their
strategy with some new Boolean functions $A,B:\{0,1\}^{vK}=\{0,1\}^{v}\times\dots\times\{0,1\}^{v}\rightarrow\{0,1\}$
which, restricted to each set of $K$ questions (i.e. $K$ cosets),
take the value 1 for exactly one element. Alice's question of $G_{KV}^{\otimes K}$
is given by $u=(u_{1},...,u_{K})$ where each $u_{i}\in\{0,1\}^{v},\,i\in[K]$
is sampled uniformly. But this is equivalent to sampling $u$ uniformly
in $\{0,1\}^{vK}$. Similarly, $z=(z_{1},...,z_{K})$ is sampled bitwise
independently, since each $z_{i}$ is. Therefore, the winning probability
is given by
\begin{equation}
\begin{aligned}\mathop{\mathbb{E}}_{\substack{u_{i},z_{i}\\
i\in[K]
}
} & \left[\sum_{\substack{h_{i}\in H_{i}\\
i\in[K]
}
}A((u_{1},...,u_{K})\oplus(h_{1},...,h_{K}))B((u_{1},...,u_{K})\oplus(z_{1},...,z_{K})\oplus(h_{1},...,h_{K}))\right]\\
 & =\sum_{\substack{h_{i}\in H_{i}\\
i\in[K]
}
}\mathop{\mathbb{E}}_{u,z}\left[A(u\oplus(h_{1},...,h_{K}))B(u\oplus z\oplus(h_{1},...,h_{K}))\right]\\
 & =v^{K}\mathop{\mathbb{E}}_{u,z}\left[A(u)B(u\oplus z)\right]
\end{aligned}
\end{equation}
since there are $v^{K}$ choices of strings of the form $(h_{1},...,h_{K})\in H_{1}\times...\times H_{K}$.
To bound $\mathbb{E}_{u,z}\left[A(u)B(u\oplus z)\right]$, we follow
the computation of Ref. \cite[Theorem 4.1]{buhrman_near-optimal_2012},
which uses the Cauchy-Schwarz and hypercontractive inequalities, and
obtain
\begin{equation}
\mathbb{E}_{u,z}\left[A(u)B(u\oplus z)\right]\leq\frac{1}{v^{K/(1-\eta)}}.
\end{equation}

If, instead, the parties share a distribution of the form
\begin{equation}
P_{1}(\alpha,\{b_{i}\}_{i\leq k}|\chi,\{y_{i}\}_{i\leq k})P_{2}(\{b_{i}\}_{i>k}|\{y_{i}\}_{i>k}),\label{eq:KVblprob}
\end{equation}
for each $\alpha\equiv(a_{1},...,a_{K}),$ $\chi\equiv(x_{1},...,x_{K}),$
$b_{i},y_{i}$, $i\in[K]$, for some $k\in[K]$, then their winning
probability is given by
\begin{equation}
\begin{aligned}\Bigl\langle\widetilde{G},P_{1}P_{2}\Bigr\rangle=\sum_{\substack{a_{i},b_{i},x_{i},y_{i}\\
i\in[K]
}
}\prod_{i=1}^{K} & G_{a_{i}b_{i}|x_{i}y_{i}}P_{1}(\alpha,\{b_{i}\}_{i\leq k}|\chi,\{y_{i}\}_{i\leq k})P_{2}(\{b_{i}\}_{i>k}|\{y_{i}\}_{i>k})\\
=\sum_{\substack{a_{i},b_{i},x_{i},y_{i}\\
k<i\leq K
}
}\prod_{i=k+1}^{K} & G_{a_{i}b_{i}|x_{i}y_{i}}\left(\sum_{\substack{a_{i},b_{i},x_{i},y_{i}\\
1\leq i\leq k
}
}\prod_{i=1}^{k}G_{a_{i}b_{i}|x_{i}y_{i}}P_{1}(\alpha,\{b_{i}\}_{i\leq k}|\chi,\{y_{i}\}_{i\leq k})\right)\\
 & \times P_{2}(\{b_{i}\}_{i>k}|\{y_{i}\}_{i>k})
\\=\sum_{\substack{a_{i},b_{i},x_{i},y_{i}\\
k<i\leq K
}
}\prod_{i=k+1}^{K} & G_{a_{i}b_{i}|x_{i}y_{i}}f(\{a_{i}\}_{i>k},\{x_{i}\}_{i>k})P_{2}(\{b_{i}\}_{i>k}|\{y_{i}\}_{i>k}),
\end{aligned}
\label{eq:KV-local}
\end{equation}
where we define
\begin{equation}
f(\{a_{i}\}_{i>k},\{x_{i}\}_{i>k})=\sum_{\substack{a_{i},b_{i},x_{i},y_{i}\\
1\leq i\leq k
}
}\prod_{i=1}^{k}G_{a_{i}b_{i}|x_{i}y_{i}}P_{1}(\alpha,\{b_{i}\}_{i\leq k}|\chi,\{y_{i}\}_{i\leq k}).
\end{equation}
We will now show that we can define a probability distribution $\tilde{P}(\{a_{i}\}_{i>k}|\{x_{i}\}_{i>k})$
all of whose components are greater than or equal to those of $f(\{a_{i}\}_{i>k},\{x_{i}\}_{i>k})$,
and use this, together with the previous parallel repetition result,
to bound $\left\langle \widetilde{G},P_{1}P_{2}\right\rangle $. First,
note that the function $f$ is pointwise positive and such that
\begin{equation}
\sum_{a_{k+1},...,a_{K}}f(\{a_{i}\}_{i>k},\{x_{i}\}_{i>k})\leq1
\end{equation}
for all $x_{k+1},...,x_{K}$. Indeed, fixing $x_{k+1},...,x_{K}$,
we have
\begin{equation}
\begin{aligned}\sum_{a_{k+1},...,a_{K}} & \sum_{\substack{a_{i},b_{i},x_{i},y_{i}\\
1\leq i\leq k
}
}\prod_{i=1}^{k}G_{a_{i}b_{i}|x_{i}y_{i}}P_{1}(\alpha,\{b_{i}\}_{i\leq k}|\chi,\{y_{i}\}_{i\leq k})\\
 & =\sum_{\substack{a_{i},b_{i},x_{i},y_{i}\\
1\leq i\leq k
}
}\prod_{i=1}^{k}G_{a_{i}b_{i}|x_{i}y_{i}}\left(\sum_{a_{k+1},...,a_{K}}P_{1}(\alpha,\{b_{i}\}_{i\leq k}|\chi,\{y_{i}\}_{i\leq k})\right)\\
 & \leq\sum_{\substack{x_{i},y_{i}\\
1\leq i\leq k
}
}\prod_{i=1}^{k}\max_{a_{i},b_{i}}G_{a_{i}b_{i}|x_{i}y_{i}}\left(\sum_{\substack{a_{i},i\in[K]\\
b_{i},i\leq k
}
}P_{1}(\alpha,\{b_{i}\}_{i\leq k}|\chi,\{y_{i}\}_{i\leq k})\right)\\
 & =\sum_{\substack{x_{i},y_{i}\\
1\leq i\leq k
}
}\prod_{i=1}^{k}\max_{a_{i},b_{i}}G_{a_{i}b_{i}|x_{i}y_{i}}\leq1,
\end{aligned}
\end{equation}
where the last inequality follows from equation (\ref{eq:KVmulti-normn}).
Thus, for each $x_{k+1},...,x_{K}$ we can define $\tilde{P}(\{a_{i}\}_{i>k}|\{x_{i}\}_{i>k})$
to have the same elements as $f$ except when all $a_{i}=0$:
\begin{equation}
\tilde{P}(\{a_{i}\}_{i>k}|\{x_{i}\}_{i>k})=\begin{cases}
f(\{a_{i}\}_{i>k},\{x_{i}\}_{i>k}) & \text{if }a_{i}\neq0\text{ for some }i>k,\\
1-\sum_{\left\{ a_{i}^{\prime}\right\} _{i>k}\neq\overrightarrow{0}}f(\{a_{i}^{\prime}\}_{i>k},\{x_{i}\}_{i>k}) & \text{if }a_{i}=0\text{ for all }i>k.
\end{cases}
\end{equation}
Then, $\tilde{P}$ is a probability distribution, all of whose components
are larger than or equal to those of $f$, and hence from equation
(\ref{eq:KV-local}) we deduce that
\begin{equation}
\left\langle \widetilde{G},P_{1}P_{2}\right\rangle \leq\sum_{\substack{a_{i},b_{i},x_{i},y_{i}\\
k<i\leq K
}
}\prod_{i=k+1}^{K}G_{a_{i}b_{i}|x_{i}y_{i}}\tilde{P}(\{a_{i}\}_{i>k}|\{x_{i}\}_{i>k})P_{2}(\{b_{i}\}_{i>k}|\{y_{i}\}_{i>k}).
\end{equation}
But the right-hand side is the winning probability of the $(K-k)$-fold
parallel repetition of $G_{KV}$ using the bilocal distribution $\tilde{P}P_{2}$
which, repeating the calculation above, can be found to be bounded
as
\begin{equation}
\left\langle \widetilde{G},P_{1}P_{2}\right\rangle \leq\frac{v^{K-k}}{v^{(K-k)/(1-\eta)}}.
\end{equation}
This exhausts the local strategies available to the players. Comparing
the bound just obtained to the one from the distribution in equation
(\ref{eq:KVparallelprob}), we find $v^{K}\geq v^{K-k}\geq v$, since
$v>1$ and $K>k$. Hence,
\begin{equation}
\frac{v^{K}}{v^{K/(1-\eta)}}\leq\frac{v^{K-k}}{v^{(K-k)/(1-\eta)}}\leq\frac{v}{v^{/(1-\eta)}}.
\end{equation}
Taking $\eta=1/2-1/\log v$, we have that, for any bilocal distribution
$P_{\text{BL}}$,
\begin{equation}
\left\langle \widetilde{G},P_{\text{BL}}\right\rangle \leq\frac{C}{v}\label{eq:KV-BLbound}
\end{equation}
for some constant $C$.
\end{proof}

We are now ready to prove item iii) in Theorem \ref{thm superactiv}. It will be entailed by the following Lemma.
\begin{lemma}
The $n$-partite PEN state over a star network (with the central party labeled by $A=A_1\cdots A_{n-1}$ and the other parties by $\{B_i\}_{i=1}^{n-1}$),
\begin{equation}
\tau(p)=\rho_{A_1B_1}(p)\otimes\bigotimes_{i=2}^{n-1}\phi^+_{A_iB_i},
\end{equation}
gives rise to GMNL by taking many copies if the isotropic state $\rho_{A_1B_1}(p)$ is entangled.
\end{lemma}
\begin{proof}
Let us fix any $p>1/(d+1)$ such that the isotropic state $\rho_{A_1B_1}(p)$ shared by $A_1$ and $B_1$ is entangled and simply write $\rho_{A_1B_1}\equiv \rho_{A_1B_1}(p)$ and $\tau\equiv \tau(p)$ for the rest of the proof. Taking $k$ copies of the star network in the statement of the lemma
is equivalent to taking a star network where Alice shares $k$ copies
of an isotropic state, $\rho_{A_1B_1}^{\otimes k}$, with Bob$_{1}$, and
$k$ copies of a $d$-dimensional maximally entangled state, $\phi_{A_iB_i}^{+\otimes k}$,
with each Bob$_{i}$, $i=2,...,n-1$. The latter state is in turn equivalent
to a $d^{k}$-dimensional maximally entangled state $\phi_{d^{k}}^{+}$.

Let us first assume that $d$ is a power of 2, therefore so is $d^k$, and consider the Khot-Visnoi game for $v=d^k$ and the corresponding generalization introduced in Lemma \ref{KV general}. We will use the superactivation result first proved in Ref. \cite{palazuelos_super-activation_2012} and extended in Ref. \cite{cavalcanti_all_2013} to show that state $\tau^{\otimes k}$ allows the parties to win the game $\widetilde{G}$ with a higher probability than if they use any bilocal strategy.

Given the structure of the game, the probability of winning $\widetilde{G}$ using the state $\tau^{\otimes k}$ is lower bounded by the product of the probabilities of winning each $G_{KV}$ with the state at each edge $i$, since the players can play every game independently. On the maximally entangled edges, the probability of winning $G_{KV}$ is bounded by equation (\ref{eq:KV-quantum-bipart}). We will obtain a similar bound for the isotropic edge. Let $\rho_{p}$ be a $d$-dimensional isotropic state with a visibility $p$ such that the state is entangled. Its entanglement fraction \cite{isotropic} is $F=\bra{\phi_{d}^{+}}\rho_{p}\ket{\phi_{d}^{+}}=p+(1-p)/d^{2}$,
which we can use to write the isotropic state in the form
\begin{equation}
\rho_{F}=F\phi_{d}^{+}+(1-F)\frac{\uno-\phi_{d}^{+}}{d^{2}-1},
\end{equation}with $F>1/d$ whenever $\rho_{F}$ is entangled. Hence, in our situation we can write
\begin{equation}
\rho_{A_1B_1}^{\otimes k}=F^{k}\phi_{d}^{+\otimes k}+\dots,\label{eq:isocopies}
\end{equation}
where $F>1/d$ and where the omitted terms are tensor products of $\phi_{d}^{+}$ and
$(\uno-\phi_{d}^{+})/(d^{2}-1)$ with coefficients that are products
of $F$ and $(1-F)$. Acting on $\rho_{A_1B_1}^{\otimes k}$ with the same
projective measurements as above gives a probability distribution
$P_{\text{iso}}$ which is linear in the terms of equation (\ref{eq:isocopies}),
and the action of $G_{KV}$ on this distribution is linear too. Since
the coefficients $G_{a_{i}b_{i}|x_{i}y_{i}}$ are nonnegative, we
have
\begin{equation}
\left\langle G_{KV},P_{\text{iso}}\right\rangle \geq F^{k}\left\langle G_{KV},P_{1}\right\rangle ,\label{eq:KV-isobound}
\end{equation}
where $P_{1}$ is the probability distribution obtained from the projective
measurements acting on $\phi_{d^{k}}^{+}$. By equation (\ref{eq:KV-quantum-bipart}),
we find
\begin{equation}
F^{k}\left\langle G_{KV},P_{1}\right\rangle \geq F^{k}(1-2\eta)^{2}.
\end{equation}
Since we have fixed $\eta=1/2-1/\log v$ like in Lemma \ref{lem:KV-bl}, we find
\begin{equation}
F^{k}\left\langle G_{KV},P_{1}\right\rangle \geq F^{L}\frac{D}{\ln^{2}v}=F^{k}\frac{D}{k^{2}\ln^{2}d},
\end{equation}
where $D$ is a universal constant.

Putting the bounds from both types of edges together, and denoting
by $P_{Q}$ the probability distribution obtained from the state of
the whole network and the projective measurements that are performed
on each edge, we obtain
\begin{equation}
\left\langle \widetilde{G},P_{Q}\right\rangle \ge F^{k}\frac{D}{k^{2}\ln^{2}d}\left(\frac{D}{k^{2}\ln^{2}d}\right)^{n-2}=\frac{F^{k}D^{n-1}}{k^{2(n-1)}\ln^{2(n-1)}d}.\label{eq:KV-quantum}
\end{equation}

Finally, we use Lemma \ref{lem:KV-bl} to compare the local and quantum
bounds. Using equations (\ref{eq:KV-BLbound}) and (\ref{eq:KV-quantum}),
we find
\begin{equation}
\frac{\left\langle \widetilde{G},P_{Q}\right\rangle }{\sup_{P_{BL}\in\mathcal{BL}}\left\langle \widetilde{G},P_{\text{BL}}\right\rangle }\geq\frac{D^{n-1}}{Ck^{2(n-1)}\ln^{2(n-1)}d}F^{k}d^{k},
\end{equation}
where $\mathcal{BL}$ is the set of bilocal distributions. Since $F>1/d$, this expression tends to $\infty$
as $k$ grows unbounded. In particular, the ratio is $>1$, proving
that GMNL is obtained.

The case where $d$ is not a power of 2 was commented in Remark 1.1 of Ref. \cite{palazuelos_largest_2014}. Here, one can modify the Khot-Visnoi game to obtaining a similar bound on the quantum winning probability but with a different constant $D$. The bound on the classical winning probability (equation (\ref{eq:KV-BLbound})) is unchanged. Hence, the same proof works in this case.
\end{proof}

\end{widetext}


\begin{thebibliography}{widest-label}

\bibitem{reviewe} R. Horodecki \textit{et al.}, Rev. Mod. Phys. \textbf{81}, 865 (2009).

\bibitem{reviewcm} L. Amico, R. Fazio, A. Osterloh, V. Vedral, Rev. Mod. Phys. {\bf 80}, 517
(2008); J. Eisert, M. Cramer, and M. B. Plenio, Rev. Mod. Phys. \textbf{82}, 277 (2010).

\bibitem{secretsharing} M. Hillery, V. Bu\v{z}ek, and A. Berthiaume, Phys. Rev. A \textbf{59}, 1829 (1999); D. Gottesman, Phys. Rev. A 61, 042311 (2000).

\bibitem{conferencekey} K. Chen and H.-K. Lo, in Proceedings. International Symposium on Information Theory, 2005. ISIT 2005, 1607 (2005); R. Augusiak and P. Horodecki, Phys. Rev. A \textbf{80}, 042307 (2009).

\bibitem{1wayqc} R. Raussendorf and H. J. Briegel, Phys. Rev. Lett. \textbf{86}, 5188 (2001).

\bibitem{qalgo} D. Bru{\ss} and C. Macchiavello, Phys. Rev. A \textbf{83}, 052313 (2011).

\bibitem{conferencekeygme} M. Epping, H. Kampermann, C. Macchiavello, and D. Bru{\ss}, New J. Phys. \textbf{19}, 093012 (2017).

\bibitem{gmesensing} P. Hyllus, W. Laskowski, R. Krischek, C. Schwemmer, W. Wieczorek, H. Weinfurter, L. Pezz\'{e}, and A. Smerzi, Phys. Rev. A \textbf{85}, 022321 (2012); G. T\'oth, Phys. Rev. A \textbf{85}, 022322 (2012).

\bibitem{GMEkey} S. Das, S. B\"auml, M. Winczewski, and K. Horodecki, Phys. Rev. X \textbf{11}, 041016 (2021).

\bibitem{gmecertification} O. G\"uhne and G. T\'oth, Phys. Rep. \textbf{474}, 1 (2009); N. Friis, G. Vitagliano, M. Malik, and M. Huber, Nat. Rev. Phys. \textbf{1}, 72 (2019).

\bibitem{gurvits} L. Gurvits, J. Comput. Syst. Sci. \textbf{69}, 448 (2004).

\bibitem{reviewn} K. Azuma, S. B\"auml, T. Coopmans, D. Elkouss, and B. Li, AVS Quantum Sci. \textbf{3}, 014101 (2021).

\bibitem{GNME1} M. Navascues, E. Wolfe, D. Rosset, and A. Pozas-Kerstjens, Phys. Rev. Lett. \textbf{125}, 240505 (2020).

\bibitem{GNME2} T. Kraft, S. Designolle, C. Ritz, N. Brunner, O. G\"uhne, and M. Huber, Phys. Rev. A \textbf{103}, L060401 (2021).

\bibitem{networkLOCC} For works studying LOCC convertibility in networks sharing particular pure bipartite entangled states, see H. Yamasaki, A. Soeda, and M. Murao, Phys. Rev. A \textbf{96}, 032330 (2017); H. Yamasaki, A. Pirker, M. Murao, W. D\"ur, and B. Kraus, Phys. Rev. A \textbf{98}, 052313 (2018); C. Spee and T. Kraft, arXiv:2105.01090 (2021).

\bibitem{networkGMNL} P. Contreras-Tejada, C. Palazuelos, and J. I. de Vicente, Phys. Rev. Lett. \textbf{126}, 040501 (2021).

\bibitem{graph} M. Hein, W. D\"ur, J. Eisert, R. Raussendorf, M. Van den Nest, and H.-J. Briegel, \textit{Proceedings of the International School of Physics "Enrico Fermi" on "Quantum Computers, Algorithms and Chaos"}, arXiv:quant-ph/0602096 (2006); H. J. Briegel, D. E. Browne, W. D\"ur, R. Raussendorf, and M. Van den Nest, Nature Physics \textbf{5}, 19 (2009).

\bibitem{LME} C. Kruszynska and B. Kraus, Phys. Rev. A \textbf{79}, 052304 (2009); M. Rossi, M. Huber, D. Bru{\ss}, and C. Macchiavello, New J. Phys. \textbf{15}, 113022 (2013).

\bibitem{tensor} R. Or\'us, Nat. Rev. Phys. \textbf{1}, 538 (2019).

\bibitem{GMEneqGMNL1} R. Augusiak, M. Demianowicz, J. Tura, and A. Acin, Phys. Rev. Lett. \textbf{115}, 030404 (2015); R. Augusiak, M. Demianowicz, and J. Tura, Phys. Rev. A \textbf{98}, 012321 (2018).

\bibitem{GMEneqGMNL2} J. Bowles, J. Francfort, M. Fillettaz, F. Hirsch, and N. Brunner, Phys. Rev. Lett. \textbf{116}, 130401 (2016).

\bibitem{palazuelos_super-activation_2012} C. Palazuelos, Phys. Rev. Lett. \textbf{109}, 190401 (2012).

\bibitem{isotropic} M. Horodecki and P Horodecki, Phys. Rev. A \textbf{59}, 4206 (1999).

\bibitem{suppmat} See supplemental material for more details on the definition of PEN states and the proofs of Theorems 1, 2 and 3 and the bounds given in Table I.

\bibitem{sunchen} Y. Sun and L. Chen, Ann. Phys. (Berlin) \textbf{533}, 2000432 (2021).

\bibitem{jungnitsch} B. Jungnitsch, T. Moroder, and O. G\"uhne, Phys. Rev. Lett. \textbf{106}, 190502 (2011).

\bibitem{1waydistillation} I. Devetak and A. Winter, Proc. R. Soc. A \textbf{461}, 207 (2005).

\bibitem{reviewnl} N. Brunner, D. Cavalcanti, S. Pironio, V. Scarani, and S. Wehner, Rev. Mod. Phys. \textbf{86}, 419 (2014).




\bibitem{GWAM12} R. Gallego, L. E. W\"urflinger, A. Acin, and M. Navascues, Phys. Rev. Lett. \textbf{109}, 070401 (2012).

\bibitem{BBGS13} J.-D. Bancal, J. Barrett, N. Gisin, and S. Pironio, Phys. Rev. A \textbf{88}, 014102 (2013).

\bibitem{SRB20} D. Schmid, D. Rosset, and F. Buscemi, Quantum \textbf{4}, 262 (2020).

\bibitem{WSSKS20} E. Wolfe, D. Schmid, A. B. Sainz, R. Kunjwal, and R. W. Spekkens, Quantum \textbf{4}, 280 (2020).


\bibitem{Werner89} R. F. Werner, Phys. Rev. A \textbf{40}, 4277 (1989).

\bibitem{Barrett02} J. Barrett, Phys. Rev. A \textbf{65}, 042302 (2002).




\bibitem{steering} R. Uola, A. C. S. Costa, H. C. Nguyen, and O. G\"uhne, Rev. Mod. Phys. \textbf{92}, 15001 (2020).

\bibitem{almeida} M. L. Almeida, S. Pironio, J. Barrett, G. Toth, and A. Acin, Phys. Rev. Lett. \textbf{99}, 040403 (2007).

\bibitem{amr_unbounded_2020} A. Amr, C. Palazuelos, and J. I. de Vicente, J. Phys. A: Math. Theor. \textbf{53}, 275301 (2020).

\bibitem{cavalcanti_all_2013} D. Cavalcanti, A. Acin, N. Brunner, and T. Vertesi, Phys. Rev. A \textbf{87}, 042104 (2013).

\bibitem{bennett_teleporting_1993} C. H. Bennett, G. Brassard, C. Crepeau, R. Jozsa, A. Peres, and W. K. Wootters, Phys. Rev. Lett. \textbf{70}, 1895 (1993).


\bibitem{bennett_mixed-state_1996}  C. H. Bennett, D. P. DiVincenzo, J. A. Smolin, and W. K. Wootters, Phys. Rev. A  \textbf{54}, 3824 (1996).

\bibitem{hamada_exponential_2002} M. Hamada, Phys. Rev. A \textbf{65}, 052305 (2002).

\bibitem{biswas_genuine-multipartite-entanglement_2014} A. Biswas, R. Prabhu, A. Sen(De), and U. Sen, Phys. Rev. A \textbf{90}, 032301 (2014).



\bibitem{svetlichny_distinguishing_1987} G. Svetlichny, Phys. Rev. D \textbf{35}, 3066 (1987).



\bibitem{khot_unique_2005} S. A. Khot and N. K. Vishnoi, 46th Annual IEEE Symposium on Foundations of Computer Science (FOCS'05), 53-62 (2005).

\bibitem{buhrman_near-optimal_2012}  H. Buhrman, O. Regev, G. Scarpa, and R. de
Wolf, Theory of Computing \textbf{8}, 623 (2012).

\bibitem{raz_parallel_1998} R Raz, SIAM J. Comput. \textbf{27}, 763 (1998).



\bibitem{palazuelos_largest_2014} C. Palazuelos, J. Funct. Anal. \textbf{267}, 1959 (2014).


\end{thebibliography}
\end{document}